\documentclass[11pt]{article} 

\usepackage[utf8]{inputenc} 

\usepackage{geometry} 
\usepackage[hidelinks]{hyperref}
\usepackage{graphicx} 
\usepackage{amsfonts,amsthm,amsmath,amssymb,esint}
\usepackage{slashed}
\usepackage{authblk}
\usepackage{tikz}
\usetikzlibrary{arrows}
\usepackage{stmaryrd}
\newcommand{\R}{\mathbb{R}}
\newcommand{\C}{\mathbb{C}}

\newcommand{\N}{\mathbb{N}}
\newcommand{\G}{\mathbb{G}}
\newcommand{\F}{\mathbb{F}}

\newcommand{\Sb}{\mathbb{S}}

\newcommand{\Hb}{\mathbb{H}}

\newcommand{\ket}{\rangle}
\newcommand{\bra}{\langle}

\newcommand{\Li}{\mathcal{L}}

\newcommand{\Ai}{\mathcal{A}}
\newcommand{\Bi}{\mathcal{B}}

\newcommand{\Hi}{\mathcal{H}}

\newcommand{\Ki}{\mathcal{K}}

\newcommand{\GH}{\mathfrak{H}}

\newcommand{\GT}{\mathfrak{T}}

\newcommand{\id}{\operatorname{id}}

\newcommand{\diag}{\operatorname{diag}}

\newcommand{\ep}{\varepsilon}

\newcommand{\Tr}{\operatorname{Tr}}

\newcommand{\SU}{\operatorname{SU}}
\newcommand{\Un}{\operatorname{U}}
\newcommand{\On}{\operatorname{O}}

\newcommand{\Mn}{\operatorname{M}}
\newcommand{\GeL}{\operatorname{GL}}

\newtheorem{thm}{Theorem}[section]

\newtheorem{Lemma}[thm]{Lemma}
\newtheorem{prop}[thm]{Proposition}
\newtheorem{cor}[thm]{Corollary}

\theoremstyle{definition}
\newtheorem{dfn}[thm]{Definition}
\newtheorem{Notation}[thm]{Notation}
\newtheorem{Remark}[thm]{Remark}
\newtheorem{Example}[thm]{Example}

\newtheorem*{Ack}{Acknowledgment}

\numberwithin{equation}{section}

\newcommand{\bone}{\mathbf{1}}

\usepackage{booktabs} 
\usepackage{array} 
\usepackage{paralist} 
\usepackage{verbatim} 
\usepackage{subfig} 

\usepackage{fancyhdr} 
\evensidemargin=\oddsidemargin
\usepackage{sectsty}
\allsectionsfont{\sffamily\mdseries\upshape} 

\usepackage[nottoc,notlof,notlot]{tocbibind} 
\usepackage[titles,subfigure]{tocloft} 

\title{Detailed balance as a quantum-group symmetry of Kraus operators}
\author{Andreas Andersson}
\affil{\small Email: fornjotnr@hotmail.com}
\affil{\footnotesize Max Planck Institute for Mathematics in the Sciences, Inselstrasse 22, D-04103 Leipzig, Germany
\\Wollongong University, School of Mathematics and Applied Statistics, 2522 Wollongong, Australia}
\affil{Mathematics Subject Classification 2010 Primary: 81S22; Secondary: 81R50, 82C10}
\affil{Keywords: Quantum channels, completely positive maps, Stinespring representation, KMS states, time reversal, compact quantum groups, Kraus operators.}

\begin{document}
\maketitle
\abstract
A unital completely positive map governing the time evolution of a quantum system is usually called a quantum channel, and it can be represented by a tuple of operators which are then referred to as the Kraus operators of the channel. We look at states of the system whose correlations with respect to the channel have a certain symmetry. Then we show that detailed balance holds if the Kraus operators satisfy a very interesting algebraic relation which plays an important role in the representation theory of any compact quantum group. 


\section*{Introduction}
Nonequilibrum fluctuation relations can be used to calculate the change in free energy in a system due to a change in externally controlled parameters \cite{Jarz1}, \cite{Jarz2}, \cite{Crook4}. These relations can be derived by comparing the work done in the process with the work done in a ``backward" or ``time-reversed" process. For a classical or a closed quantum system it is rather clear how to define such a time reversal. There is however hope that analogous fluctuation relations may hold also for dissipative quantum systems \cite{Rast1}, \cite{ALMZ1}, \cite{RaZy1}. Therefore, a suitable notion of time reversal is needed also for open quantum systems.

As is standard, we will consider the setting where the time evolution from some initial time $t=0$ up to the next time of interest $t=\tau$ is completely positive and probability preserving. In the Heisenberg picture (in which observables are time dependent, while states are time independent), we are thus discussing a unital completely positive map $\Phi:\Bi(\Hi_0)\to\Bi(\Hi_0)$ (a quantum channel, for short) on the algebra $\Bi(\Hi_0)$ of all bounded observables on the system Hilbert space $\Hi_0$ \cite{AlLe1}. Every quantum channel on $\Bi(\Hi_0)$ has a ``Kraus representation"
$$
\Phi(A)=\sum^n_{k=1}K_k^*AK_k,\qquad \forall A\in\Bi(\Hi_0),
$$
for some $n\in\N\cup\{\infty\}$ and some bounded operators $K_k$ on $\Hi_0$. We shall assume $n<\infty$ in this paper. The unitality $\Phi(\bone)=\bone$ reads $\sum^n_{k=1}K_k^*K_k=\bone$ in terms of the Kraus operators. 

We shall need the whole discrete semigroup $(\Phi^m)_{m\in\N_0}$ obtained from a single channel $\Phi$. Each quantum channel $\Phi^m:=\Phi\circ\cdots\circ\Phi$ gives rise to a representation of $\Bi(\Hi_0)$ on a Hilbert space of the form $\Hi_0\otimes\GH_m$, known as the Stinespring representation. 
We shall take the sequence $(\GH_m)_{m\in\N_0}$ as a $\Phi$-dependent model for the environment. This ``Stinespring bath" does not describe a system that exists on its own but rather models the interaction between the system of interest with its total environment. The choice of this bath is based on two desiderata:
\begin{enumerate}[(i)]
\item{We would like it to be possible to deduce the properties of the dynamics directly from the relations among the Kraus operators $K_1,\dots,K_n$.}
\item{We would like the bath model to have a nice mathematical structure which says something about the channel $\Phi$.}
\end{enumerate}
We shall see that (i) and (ii) go hand in hand. It is not enough to consider the Stinespring representation of $\Phi$ alone; it gives only a first approximation for the environment. It is the structure of the whole sequence $(\GH_m)_{m\in\N_0}$ which is the key.

For open systems, time reversal is tightly connected to the property of detailed balance \cite{Agar1}, \cite{FaRe1}, \cite{KFGV}. Both notions are usually defined with respect to a fixed density matrix $\rho_0$ acting on the system Hilbert space $\Hi_0$. The strategy here will be to first deduce what properties the density matrix $\rho_0$ has to satisfy, given that the time evolution is governed by the channel $\Phi$. Then we deduce the minimal algebraic relations required to be satisfied by the Kraus operators for detailed balance to hold with respect to such a state $\rho_0$. The density matrix $\rho_0$ plays the role of asymptotic equilibrium state, and thermodynamical quantities can be defined in terms of $\rho_0$. However, this thermodynamics is \emph{with respect to} the channel $\Phi$ and that is all that is required. In this approach, $\rho_0$ could even be a pure state on $\Bi(\Hi_0)$. 

The algebraic relations imposed on the Kraus operators, supposed to capture the time-reversal properties of the system, ensure that the quantum channel $\Phi$ has a certain quantum symmetry. Ordinary symmetry groups can only be obtained for commuting Kraus operators, which in a sense corresponds to no heat dissipation. When we say that $\Phi$ is symmetric under some compact group $G$ we basically mean ``covariance" of the quantum channel $\Phi$ under $G$ (cf. \cite{Janz1}), and this notion extends to ``compact quantum groups". We shall not need any quantum-group theory in deriving the main results of the paper. However, for the purpose of interpretation it is often easier if there is some well-defined notion of symmetry in the background. Moreover, the assumption of quantum-group symmetry gives further structure to the Stinepring bath model $(\GH_m)_{m\in\N_0}$ that will be used in subsequent applications. Therefore, for useful reference we shall discuss how quantum groups appear already in this paper, since the main reason why they do so concerns time reversal and detailed balance. 

There are some very interesting applications of quantum groups known in the physics literature \cite{Rayc1}, \cite{Bon1}, \cite{BoDa1}, \cite{BKLRRT}, \cite{KiNe1}, \cite{KiNe3}, \cite{ADF1}, \cite{AnFr2}, but all of them involve $q$-deformations of classical Lie groups, specifically $\SU(n)$ and $\operatorname{Sp}(n)$. These phenomenologically discovered models involving $\SU_q(n)$ or $\operatorname{Sp}_q(n)$ can successfully reproduce experimental data in the same way as their classical analogues in more idealistic physical systems (e.g. one modeled by the harmonic oscillator) \cite{FIL1}. The general interpretations of §\ref{interpretsec} apply to these quantum groups as well.

\begin{Ack}
The author thanks Adam Rennie for his nice comments about the paper. 
\end{Ack}

\section{Background}

\subsection{Completely positive maps in quantum physics}
In this paper it will be useful to remember how the time evolution of a quantum system can be obtained by tracing out an environmental degree of freedom (see e.g. \cite[§8.2]{AlFa1}). 

Suppose that $\Hi_0$ is the Hilbert space of a quantum system (briefly, the ``system") that we are interested in, and that $\Ki$ is a Hilbert space of another quantum system (the ``bath") interacting with it. For example, $\Hi_0$ could be the state space of a energy levels in a molecule and $\Ki$ the state space of the bath of vibration modes capable of exchanging phonons with the molecule. All quantum open-system applications use the same setting. The dimension of $\Hi_0$ and $\Ki$ could be finite or infinite but we shall assume all Hilbert spaces to have a countable basis.

Suppose that $\rho$ and $\sigma$ are density matrices describing the states of the system and the bath, respectively. The total system $\Hi_0\otimes\Ki$ is assumed to evolve from $t=0$ to some later time $t=\tau$ via a unitary operator $W\in\Bi(\Hi_0\otimes\Ki)$. The state of the system is obtained by tracing over the bath degrees of freedom, so that after the evolution $W$ has been applied, we get the system state
\begin{equation}\label{traceoverPhi}
\Phi_*(\rho):=(\id\otimes\Tr_\Ki)\big(W(\rho\otimes\sigma)W^*\big).
\end{equation}
The evolution of the system alone therefore defines a map 
$$
\Phi_*:\Bi(\Hi_0)_*\to\Bi(\Hi_0)_*
$$
on the algebra $\Bi(\Hi)_*=\Li^1(\Hi_0)$ of trace-class operators. Writing $\sigma$ in terms of an orthonormal basis $e_1,\dots,e_n$ for $\Ki$ (where, as mentioned, the dimension $n$ of $\Ki$ could be finite or infinite) as
$$
\sigma=\sum_{k=1}^n\sigma_k|e_k\ket\bra e_k|,
$$
where $0\leq\sigma_k\leq 1$, we have
$$
\Phi_*(\rho)=\sum_{j,k=1}^nK_{j,k}\rho K_{j,k}^*,
$$
with
\begin{equation}\label{defKrausops}
K_{j,k}:=\sqrt{\sigma_k}\bra e_j|We_k\ket.
\end{equation}
Relabeling the operators $K_{j,k}$ using only one index $k=1,\dots,d:=n^2$, we write
$$
\Phi_*(\rho)=\sum_{k=1}^dK_{k}\rho K_{k}^*,\qquad\forall\rho\in\Bi(\Hi_0)_*
$$ 
and call the $K_k$'s a set of \textbf{Kraus operators} for $\Phi_*$. The map $\Phi:\Bi(\Hi_0)\to\Bi(\Hi_0)$ defined via the relation
$$
\Tr\big(\Phi(A)\rho\big)=\Tr\big(A\Phi_*(\rho)\big),\qquad\forall A\in\Bi(\Hi_0),\,\rho\in\Bi(\Hi_0)_*
$$
takes the form
\begin{equation}\label{HeisredPhi}
\Phi(A)=(\id\otimes\Tr_\Ki)\big((\bone\otimes\sigma)W^*(A\otimes\bone)W\big)=\sum_{k=1}^dK_{k}^*A K_{k}
\end{equation}
and describes the Heisenberg evolution of the system, i.e. the time evolution of observables (in contrast to the \emph{state} transformer $\Phi_*$, that is). 

Both $\Phi_*$ and $\Phi$ are completely positive maps. Furthermore, since $W$ is an isometry and $\Tr(\sigma)=1$, the Kraus operators satisfy (with convergence in the weak operator topology if the dimensions are inifinite)
\begin{equation}\label{unital}
\sum_{k=1}^dK_{k}^*K_{k}=\bone,
\end{equation}
which says that $\Phi_*$ preserves the trace and that $\Phi$ is unital. From the way the Kraus operators are defined in \eqref{defKrausops} however, even though $W$ is a coisomery, it is not true in general that
\begin{equation}\label{tracepres}
\sum_{k=1}^dK_{k}K_{k}^*=\bone.
\end{equation}
Condition \eqref{tracepres} says precisely that $\Phi$ preserves the trace when restricted to trace-class operators and that $\Phi_*$ is unital\footnote{Even though the identity operator is not trace-class when $\Hi_0$ is infinite-dimensional, when condition \eqref{tracepres} holds we can extend $\Phi_*$ to the algebra of all bounded operators using the Kraus representation.}.

\begin{dfn}\label{Defqchannel}
A completely positive map $\Phi:\Bi(\Hi_0)\to\Bi(\Hi_0)$ is a \textbf{quantum operation} if $\Phi(\bone)\leq\bone$. It is a \textbf{quantum channel} (or \textbf{stochastic map}) if $\Phi(\bone)=\bone$ and \textbf{bistochastic} if $\Phi(\bone)=\bone=\Phi_*(\bone)$. 
\end{dfn}
The relevance of quantum operations is that in general it might not be a valid approximation to regard the system $\Hi_0\otimes\Ki$ as being closed, and some probability might be lost. One could also imaging that a ``selection process" is taking place after the evolution $\rho\to\Phi_*(\rho)$ so that only certain states are maintained in the description. In any case, such a selection process should not lead to a probability increase, i.e. we should have
$$
\Tr\big(\rho\Phi(\bone)\big)=\Tr\big(\Phi_*(\rho)\big)\leq\Tr(\rho),
$$
which is the motivation for the notion of quantum operation as in Definition \ref{Defqchannel}. However, here we shall regard $\Hi_0\otimes\Ki$ as closed and work with quantum \emph{channels} only. The bistochastic operations are very special however, and not what one gets from system-bath models unless the bath is at infinite or (possibly) zero temperature. They have applications though, including the ones that can be put in a ``random unitary" form $\Phi(A)=\sum_k\lambda_k U_k^*AU_k$ with unitaries $U_k$ and (not necessarily positive) real numbers $\lambda_k$ with $\sum_k\lambda_k=1$ (see e.g. \cite{MeWo1}). 

\begin{Remark}[Commuting Kraus operators]\label{commuteKrausremark}
Suppose further that $W$ is the interaction describing a measurement of some jointly measurable selfadjoint operators $A_1,\dots,A_N$ acting in $\Hi_0$, i.e. $[A_\mu,A_\nu]=0$ for all $\mu,\nu$ and 
$$
W=\exp\Big(-i\sum_{\mu=1}^NA_\mu\otimes B_\mu\Big)
$$
for some commuting selfadjoint operators $B_1,\dots,B_N$ acting in $\Ki$. Then with the convention $\lambda^\mu A_\mu:=\sum_{\mu=1}^N\lambda^\mu A_\mu$ for $\lambda=(\lambda^1,\dots,\lambda^N)\in\R^N$ etc., and writing (see \cite{An1})
$$
W=\int_{\R^N}e^{-i\lambda^\mu A_\mu}\otimes\,dE^{\mathbf{B}}(\lambda),
$$
where $E^{\mathbf{B}}$ is the joint spectral measure of the $B_\mu$'s, we get 
$$
K_{j,k}=\bra e_j|e^{-iA_\mu\otimes B^\mu}e_k\ket=\int_{\R^N}\bra e_j|dE^{\mathbf{B}}(\lambda)e_k\ket e^{-i\lambda^\mu A_\mu},
$$
from which it is clear that $[K_{j,k},K_{l,m}]=0$ for all $j,k,l,m$, i.e. the resulting Kraus operators mutually commute. 
\end{Remark}
\begin{Remark}[Nonommuting Kraus operators]\label{NCKrausremark}
As soon as we have two competing interactions, the Kraus operators do not commute. For example, two measurements of noncommuting observables can be described with noncommuting Kraus operators for the whole process. Similarly, if the intrinsic time evolution is not neglected, the Kraus operators at different time points will not commute unless they all commute with the system Hamiltonian. For example, if the density matrix commutes with the Hamiltonian then a phase-damping interaction with the bath is modeled by commuting Kraus operators whereas population loss must involve noncommuting ones (see e.g. \cite{LOMI}).  Starting with a total Hamiltonian $H=H_0\otimes\bone+\bone\otimes H_{\text{bath}}+H_{\text{int}}$ acting in $\Hi_0\otimes\Ki$, the channel $\Phi$ obtained as in \eqref{traceoverPhi} using $W=e^{-itH}$ for some $t$ has Kraus operators \eqref{defKrausops} which do not commute (unless $[H_0,H_{\text{int}}]=0$, which is a trivial case).
\end{Remark}

We shall only work with quantum operations $\Phi$ allowing a Kraus representation with a \emph{finite} number of Kraus operators. Finitely many Kraus operators means that the bath behaves, with respect to the interaction that $\Phi$ corresponds to, like a finite-dimensional quantum system. So we may take the bath Hilbert space to be finite-dimensional and we use the symbol $\GH$ for the finite-dimensional bath. As before, let $n$ the dimension of the bath, so $\GH\cong\C^n$. Then we may identify $\Bi(\GH)$ with $\Mn_n(\C)$ after choosing a basis for $\GH$, and $W\in\Bi(\Hi_0)\otimes\Mn_n(\C)$ is given by an $n\times n$ matrix $\{W_{j,k}\}_{j,k}^n$ of operators $W_{j,k}\in\Bi(\Hi_0)$. 

Moreover, we will always use a particular choice of bath state $\sigma$, namely the pure state 
\begin{equation}\label{purebathstate}
\sigma=|e_1\ket\bra e_1|
\end{equation}
defined by the first member of an orthonormal basis $e_1,\dots,e_n$ for $\GH$. Doing so we obtain only $n$ Kraus operators
\begin{equation}\label{firstcolumnKraus}
K_k=W_{k,1}:=\bra e_k|We_1\ket,
\end{equation}
corresponding to the first column of $W=\{W_{j,k}\}_{j,k=1}^n\in\Bi(\Hi)\otimes\Mn_n(\C)$. The properties of the resulting quantum operation,
\begin{equation}\label{nKrausrep}
\Phi(A)=\sum^n_{k=1}K_k^*AK_k,
\end{equation}
are the same as for general $\sigma$: it is always a channel, i.e. $\Phi(\bone)=\bone$, but not necessarily trace-preserving since \eqref{tracepres} does not hold. Since we are anyway ignoring the details of the bath when using $\Phi$, the first-column choice is natural in that it gives a ``minimal" description of the dynamics, which will be related to the (minimal) Stinespring representation in the next section. Moreover, any quantum channel $\Phi$ on $\Bi(\Hi_0)$ can be obtained in this way by choosing the Kraus operators to correspond to the first column of a suitable unitary matrix $W$ with entries in $\Bi(\Hi_0)$ \cite[§10.3]{Kum1}, see §\ref{Stinesection} below.

\begin{Remark}[Column versus row]\label{colrowremark}
Writing
$$
W=\sum^n_{j,k=1}W_{j,k}\otimes|e_j\ket\bra e_k|
$$
one obtains \eqref{firstcolumnKraus} from 
$$
\Phi(A)=(\id\otimes\Tr_\GH)\big((\bone\otimes|e_1\ket\bra e_1|)W^{-1}(A\otimes\bone)W\big)=\bra e_1|W^{-1}(A\otimes\bone)We_1\ket\qquad (\text{first column})
$$
where we are tracing out the bath in the density matrix $\sigma=|e_1\ket\bra e_1|$, just as in \eqref{traceoverPhi}. Another possibility is to use
$$
\Phi(A):=(\id\otimes\Tr_\GH)(W^{-1}(A\otimes |e_1\ket\bra e_1|)W)\qquad (\text{first row}).
$$
Then $\Phi$ has Kraus operators $K_j=W_{1,j}$ belonging to the first row of $W$, so the two alternatives are obtained from each other by taking the transpose of $W$. In the first-row case, the property $\Phi(\bone)=\bone$ reads $\sum_kK_kK_k^*=\bone$. We shall use the first-column choice so that it is $W$ and not the transpose $W^t$ which represents the unitary evolution. 
\end{Remark}

\begin{Remark}[Discrete time evolution]\label{Nonmarkremark}
The time evolution of many physical systems is not well approximated by the Markovian evolution given by completely positive maps representing the semigroup $\R_+$. Nevertheless, the evolution from the initial state to the final state for the time interval of interest is almost always completely positive. By discretizing time into sufficiently large time intervals, one can describe the time evolution in terms of a single map $\Phi$ and its powers $\Phi^m$. In that sense, semigroups of the form $(\Phi^m)_{m\in\N_0}$ have a much wider applicability than semigroups of the form $(\Phi_t)_{t\in\R_+}$, allowing for some non-Markovianity in continuous time. 
\end{Remark}

\begin{Remark}[Measurement interpretation]\label{POVMrem}
Due to the probability-preservation condition \eqref{unital}, the Kraus operators $K_1,\dots,K_n$ of a quantum channel $\Phi$ can be used to construct a \textbf{POVM} (positive operator-valued measure), viz. the collection $E_1,\dots,E_n$ of positive operators
$$
E_k:=K_k^*K_k
$$
satisfying $\sum_kE_k=\bone$, and can be interpreted as describing a quantum measurement (see e.g. \cite[§2.4]{BP}). In the early days of quantum mechanics, a measurement was taken to be given by projections $P_1,\dots,P_n$ satisfying $\sum_kP_k=\bone$. This implies $P_jP_k=\delta_{j,k}P_k$ and the outcome of the measurement is unambiguous. The generalization obtained using the $E_k$'s gives the possibility that the off-diagonal elements of the density matrix are not deleted in one step. This is particularly interesting when we discuss repeated measurement, as one may follow the discrete time evolution of both off-diagonal and diagonal elements. Eventually the off-diagonals may disappear completely, but the time dependences of finite-time correlations (coherence effects) are usually the interesting data.
\end{Remark}

\begin{Remark}[Fundamental unitaries and symmetries] It was realized by Ojima \cite{Ojim1} that the interaction unitary $W\in\Bi(\Hi_0\otimes\Ki)$ of a von Neumann measurement of a set $A_1,\dots,A_N$ of commuting selfadjoint operators gives a representation on $\Hi_0$ of the so-called multiplicative unitary of a locally compact group $G$ (without losing anything relevant, $G=\R^N$). The fundamental unitary implements the Hopf-algebraic coproduct on the algebra $C(G)$ of continuous functions on $G$. Anticipating from this identification that the ``twisting of coproduct" (which has been a popular way of deforming algebras) would be relevant to the operator deformations resulting from ``disturbance by measurement", the post-measurement algebra is expected to have an explicit description as a deformation of the pre-measurement one (see \cite{An1}). 
Now comes the question whether a fundamental unitary can be used also e.g. when not neglecting the intrinsic evolution of the system during the interaction with a bath (and more generally, when the Kraus operators of the quantum operation do not mutually commute). The interpretation given in this paper together with the results of \cite{An6} answers this question postively, provided that we use the fundamental unitary of a quantum group instead of restricting ourselves to classical groups. 
\end{Remark}

\subsection{Stinespring}\label{Stinesection}
The most fundamental result about completely positive maps is the Stinespring theorem.
\begin{thm}[{Stinespring's theorem, \cite[Thm. 1.1.1]{Ar7}}] Let $\Hi_0$ be a Hilbert space. Then every quantum channel $\Phi:\Bi(\Hi_0)\to\Bi(\Hi_0)$ has the form
$$
\Phi(A)=V^*\pi(A) V,\qquad \forall A\in\Bi(\Hi_0),
$$
where $V:\Hi_0\to\Hi_1$ is an embedding of $\Hi_0$ into a larger Hilbert space $\Hi_1$ and $\pi:\Bi(\Hi_0)\to\Bi(\Hi_1)$ is a map such that $\pi(AB)=\pi(A)\pi(B)$ and $\pi(A^*)=\pi(A)^*$.  
\end{thm}
There is a ``minimal" such triple $(V,\Hi_1,\pi)$ (essentially unique), and we refer to it as the \textbf{Stinespring representation} of $\Phi$.

If we know a Kraus decomposition $K_1,\dots,K_n$ of $\Phi$ as in \eqref{nKrausrep}, we can construct the space $\Hi_1$ and the maps $V$, $\pi$ as follows \cite{Ar3}. If (and only if) $n$ is the minimal number of Kraus operators that can be used to represent $\Phi$, then $K_1,\dots,K_n$ are linearly independent; let us assume that $n$ is minimal. Then $K_1,\dots,K_n$ span an $n$-dimensional vector space. We let $\GH\cong\C^n$ be the Hilbert space obtained by taking the (adjoints of the) Kraus operators as orthonormal basis vectors. That is, the inner product $\bra\cdot|\cdot\ket$ in $\GH$ is characterized by
\begin{equation}\label{innerproduonStine}
\bra K_j^*|K_k^*\ket:=\delta_{j,k},\qquad \forall j,k=1,\dots,n
\end{equation}
and every element $\phi$ of $\GH$ is of the form $\phi=\sum^n_{k=1}\phi_kK_k^*$ for some numbers $\phi_k\in\C$. We set 
$$
\Hi_1:=\Hi_0\otimes\GH,
$$
and define the isometry $V:\Hi_0\to\Hi_1$ by
$$
V\xi:=\sum^n_{k=1}K_k\xi\otimes K_k^*,\qquad \forall \xi\in\Hi_0,
$$
while the map $\pi$ is just
$$
\pi(A):=A\otimes\bone,\qquad \forall A\in\Bi(\Hi_0).
$$
\begin{Notation}
We will usually write the adjoint Kraus operators as $e_1,\dots,e_n$ when regarded as a basis for $\GH$, while we write $K_1^*,\dots,K_n^*$ for them as elements of $\Bi(\Hi_0)$. 
\end{Notation}
Now let $W\in\Bi(\Hi_0\otimes\GH)$ be a unitary whose first column coincides with $V$. Then \eqref{traceoverPhi} holds with $\sigma=|e_1\ket\bra e_1|$. The Stinesepring representation thus also gives a unitary model for $\Phi$.

It seems a little bit too good to be true that, starting from the map $\Phi$, the unknown bath interaction which corresponds to the dissipative effects can be modeled by the finite degrees of freedom $\GH\cong\C^n$. Indeed, $\Hi_0\otimes\GH$ can only capture the one-step evolution $A\to\Phi(A)$, whereas the multi-step evolution $A\to\Phi^2(A)$ and so on have no unitary model on $\Hi_0\otimes\GH$. Instead we need the Stinespring space $\GH_2$ of $\Phi^2$, which is in general a subspace of $\GH^{\otimes 2}$. In general, we have
\begin{equation}\label{StinespaceGHm}
\GH_m\subseteq\GH^{\otimes m},\qquad m\in\N
\end{equation}
if $\GH_m$ denotes the span of products $K_{j_m}^*\cdots K_{j_1}^*$ of $m$ elements of the $K_j^*$'s, for all $j_1,\dots,j_m\in\{1,\dots,m\}$. The result is that a model for the bath up to time $m\in\N_0$ is provided by $\GH_m$, and there is a unitary operator $W_m$ (which is the restriction of $W^{\otimes m}$ to $\GH_m$), such that $\Phi^m$ is obtained as
\begin{equation}\label{timemtraceover}
\Phi^m(A)=\bra e_1^{(m)}|W_m^*(A\otimes\bone)W_me_1^{(m)}\ket,\qquad \forall A\in\Bi(\Hi_0)
\end{equation}
where $e_1^{(m)}$ is the $m$th tensor power of $e_1$ regarded as an element of $\GH_m$. 

The reason why we have an inclusion $\GH_m\subseteq\GH^{\otimes m}$ is that the space $\GH_m$ is spanned by sums of products $K_{j_1}\cdots K_{j_m}$ of $m$ Kraus operators. Possible relations among the Kraus operators allows us to delete some vectors in $\GH^{\otimes m}$ to obtain $\GH_m$. For instance, if $[K_j,K_k]=0$ then for the basis vectors $e_j,e_k\in\GH$ we have
$$
e_j\otimes e_k-e_k\otimes e_j=0
$$
in the space $\GH_2$ which models $\Phi^2$. In the case $[K_j,K_k]=0$ for all $j,k$, we therefore get that $\GH_m$ is the symmetric (``Bosonic") subspace of $\GH^{\otimes m}$.
\begin{Notation}
We write a word $\mathbf{k}$ in $\{1,\dots,n\}$ as $\mathbf{k}=k_1\cdots k_m$ and refer to $|\mathbf{k}|:=m$ as the \textbf{length} of $\mathbf{k}$. We then write
$$
K_\mathbf{k}:=K_{k_1}\cdots K_{k_m},\qquad K_\mathbf{k}^*:=(K_\mathbf{k})^*=K_{k_m}^*\cdots K_{k_1}^*,
$$
$$
e_\mathbf{k}:=e_{k_1}\otimes \cdots \otimes e_{k_m},
$$
and similarly for other quantities. 
\end{Notation}
\begin{Remark}\label{nonminimalremark}
The products $K_\mathbf{j}:=K_{j_1}\cdots K_{j_m}$ do not necessarily form an orthonormal basis for $\GH_m$, nor do they provide a minimal Kraus representation for $\Phi^m$. They are excessively many but they are convenient for comparison of $\Phi^m$ with $\Phi^l$.
\end{Remark}
The sequence $(\GH_m)_{m\in\N_0}$ has a very important property, namely that \cite{Ar3}
\begin{equation}\label{subprodeq}
\GH_{m+l}\subseteq\GH_m\otimes\GH_l,\qquad \forall m,l\in\N_0,
\end{equation}
saying that $\GH_\bullet=(\GH_m)_{m\in\N_0}$ is a ``subproduct system" \cite{ShSo1}. It is \eqref{subprodeq} which allows us to talk about the large-time limit of the system (see \cite{An5}, \cite{An6}) and $\GH_\bullet$ encodes (most of) the properties of the channel $\Phi$. While there are many different Kraus decomposition of the same channel $\Phi$, the object $\GH_\bullet$ is evidently independent of the choice of such a decomposition.

\section{Symmetries needed for detailed balance}

\subsection{Minimal Kraus relations}

We have seen in Remarks \ref{commuteKrausremark} and \ref{NCKrausremark} how commuting and noncommuting Kraus operators may appear in the reduced evolutions due to different kinds of interactions. We now try to see if the algebraic relations satisfied by the Kraus operators reflect some ``symmetry" which could be used in the study of the semigroup $(\Phi^m)_{m\in\N_0}$ (whose motivation was sketched in Remark \ref{Nonmarkremark}).

The Heisenberg evolution is obtained as in Remark \ref{colrowremark} by tracing over the unitary $W\in\Bi(\Hi_0\otimes\GH)$,
\begin{equation}\label{HeisPhifirstrow}
\Phi(A)=(\id\otimes\Tr_\GH)\big((\bone\otimes |e_1\ket\bra e_1|)W^{-1}(A\otimes\bone)W\big)=\sum^n_{k=1}K_k^*AK_k.
\end{equation}
The $K_k$'s are then exactly the elements of the first column of $W$ viewed as an $n\times n$ matrix with entries in $\Bi(\Hi_0)$. 

Since we want the predual $\Phi_*$ to preserve the trace of density matrices, we want $\Phi(\bone)=\bone$. However, we do not want the tracial state to be invariant under $\Phi_*$ in general. For instance, that should not be the case if the system is driven into an equilibrium state at finite temperature. Thus, in general,
\begin{equation}\label{minimalrels}
\bone=\Phi(\bone)=\sum^n_{k=1}K_k^*K_k,\qquad \bone\ne \Phi_*(\bone)=\sum^n_{k=1}K_kK_k^*.
\end{equation}
These two relations are the most basic ones we can impose on the Kraus operators, corresponding to the most basis properties of the time evolution. We can incorporate some further detail in the second one by specifying that 
\begin{equation}\label{choiceminimalrels}
\sum^n_{k=1}Q_{k,k}^{-1}K_kK_k^*=\bone
\end{equation}
for some positive scalars $Q_{k,k}$ (where we take the inverse $Q_{k,k}^{-1}$ just to avoid that in other formulas). The number $1/Q_{k,k}$ has the interpretation as the rate (or unnormalized probability) of the occurrence of the ``subchannel" evolution $\rho\to K_k\rho K_k^*$ of the density matrix $\rho$. If all the subchannels are equally probable then $Q_{k,k}=1$ for all $k$ and the tracial state defined (if well-defined) by the identity operator $\bone$ on $\Hi_0$ is invariant under $\Phi_*$.

\subsection{The correlation matrix}
In the following we have, as always, fixed a quantum channel $\Phi$ and denote by $\GH_m$ the Hilbert spaces appearing in \eqref{StinespaceGHm}. The following object is of interest in view of the measurement interpretation given in Remark \ref{POVMrem}.
\begin{dfn}\label{corrmatrixdef}
Let $\rho_0$ be a fixed density matrix on $\Bi(\Hi_0)$. The \textbf{correlation matrix} of the state $\rho_0$ on $\Hi_0$ is the matrix $Q/\Tr(Q)\in\Bi(\GH)$ where $Q$ has entries $Q_{j,k}$ defined by (see \cite[§5.1]{Schu1})
\begin{equation}\label{corrmatrixone}
\frac{Q_{j,k}}{\Tr(Q)}:=\Tr(K_j\rho_0K_k^*),
\end{equation}
where we can fix $\Tr(Q)$ by requiring that $\Tr(Q^{-1})=\Tr(Q)$. 
\end{dfn}
It is \eqref{unital} which ensures that the correlation matrix is a density matrix. The diagonal elements of $Q/\Tr(Q)$ give the probabilities of the different measurement outcomes, while the role of an off-diagonal entry $Q_{j,k}/\Tr(Q)$ is to capture the possibility of misinterpreting outcome $j$ as being $k$ (cf. \cite[§10.1]{AlFa1}). The reason why we choose to separate the correlation matrix into a non-density matrix $Q$ and a normalization factor $\Tr(Q)$ is that the diagonal entries of $Q$ will be the numbers appearing in \eqref{choiceminimalrels}, and this may or may not lead to $\Tr(Q)=1$. 

\begin{Remark} An example of an $n\times n$ density matrix which is often written in the form $Q/\Tr(Q)$ with $\Tr(Q^{-1})=\Tr(Q)$ is the Gibbs state of a single spin-$(n-1)/2$ in a constant magnetic field of magnitude $B$. Namely, with all physical constants set to unity, the eigenvalues of the Hamiltonian $H$ are $-(n-1)B/2,-(n-3)B/2,\dots,+(n-1)B/2$ so the Gibbs state at temperature $1/\beta$ is 
$$
\frac{Q}{\Tr(Q)}=\frac{e^{-\beta H}}{\Tr(e^{-\beta H})},
$$
where the normalizing factor is a geometric series
\begin{align*}
\Tr(e^{-\beta H})&=e^{\beta (n-1)B/2}+e^{\beta (n-3)B/2}+\cdots +e^{-\beta (n-1)B/2}=\frac{e^{\beta nB/2}-e^{-\beta nB/2}}{e^{\beta B/2}-e^{-\beta B/2}}
\end{align*}
which satisfies $\Tr(e^{-\beta H})=\Tr(e^{\beta H})$.
\end{Remark}
It seems natural to assume that 
$$
K_k\rho_0K_k^*\ne 0,\qquad \forall k=1,\dots,n,
$$ 
in which case $Q$ is invertible (even if $\rho_0$ is not faithful). 

\begin{Example} If $Q=\bone$ then all outcomes $1,\dots,n$ are equally probable. This does \emph{not} require $\rho_0$ to be maximally mixed (if $\Hi_0$ is infinite-dimensional, there is no such state). 
\end{Example}

Usually, the condition of detailed balance of a dynamics is formulated with respect to a state which is (often a unique) invariant state, so a kind of asymptotic equilibrium state. Also in our approach, the relation of the state $\rho_0$ to $\Phi$ will be far from random. It turns out that we shall need $\rho_0$ to be of the following kind. 

\begin{dfn}\label{assumptonrhonoll} 
Let $p_m:\GH^{\otimes m}\to\GH_m$ be the projection. A density matrix $\rho_0$ on $\Hi_0$ has \textbf{$\Phi$-symmetric correlations} if $K_k\rho_0K_k^*\ne 0$ for all $k=1,\dots,n$ and its correlation matrix $Q$ defined by \eqref{corrmatrixone} satisfies, for all $m\in\N_0$,
\begin{equation}\label{corrmatrix}
\frac{Q_{\mathbf{j},\mathbf{k}}}{\Tr(Q_m)}=\Tr(K_\mathbf{j}\rho_0K_\mathbf{k}^*).
\end{equation}
where we write 
$$
Q_m:=p_mQ^{\otimes m}p_m,\qquad Q_{\mathbf{j},\mathbf{k}}:=\bra e_\mathbf{j}|Q_me_\mathbf{k}\ket,
$$
and (importantly) $Q^{\otimes m}$ is assumed to preserve the subspace $\GH_m\subseteq\GH^{\otimes m}$ (equivalently, $Q_m=Q^{\otimes m}p_m$).
\end{dfn}
Similarly we write $Q^{\mathbf{j},\mathbf{k}}:=\bra e_\mathbf{j}|Q_m^{-1}e_\mathbf{k}\ket$ for brevity.

For a measurement, a particular choice of Kraus representation of the associated channel may be preferred. The correlation matrix clearly depends on the choice of Kraus operators. We prefer to choose the Kraus operators such that the correlation matrix becomes as simple as possible.
\begin{Lemma}[{\cite[§III.B]{NCSB1}}]\label{LemmaChoidiago}
Let $Q$ be the $\Phi$-correlation matrix of any density matrix $\rho_0$. Then the Kraus operators of $\Phi$ can be chosen such that $Q$ is diagonal.
\end{Lemma}
\begin{Remark}\label{GNSremark}
We can deduce a bit more than Lemma \ref{LemmaChoidiago} using ideas of \cite{Stor1}, in particular \cite[Thm. 4]{Stor1}. Consider the GNS representation $\pi_{\rho_0}:\Bi(\Hi_0)\to\Bi(\Hi_{\rho_0})$ of $\Bi(\Hi_0)$ associated with $\rho_0$. Let 
$\Omega$ be the cyclic implementing vector in the GNS space $\Hi_{\rho_0}$, 
$$
\Tr(\rho_0A)=\bra\Omega|\pi_{\rho_0}(A)\Omega\ket_{\rho_0},\qquad \forall A\in\Bi(\Hi_0),
$$
and let $p_0=|\Omega\ket\bra\Omega|$ be the projection onto the line spanned by $\Omega$. Define the ``Choi matrix" of $\Phi$ (with respect to $\rho_0$) by
$$
C_\Phi:=\pi_{\rho_0}(\Phi)p_0\in\Bi(\Hi_{\rho_0})
$$
where $\pi_{\rho_0}(\Phi)$ is the operator on $\Hi_{\rho_0}$ specified by $\pi_{\rho_0}(\Phi)\pi_{\rho_0}(A)\Omega:=\pi_{\rho_0}(\Phi(A))\Omega$ for all $A\in\Bi(\Hi_0)$. 
The fact that $\Phi$ has finitely many Kraus operators implies that $C_\Phi$ has finite rank, and complete positivity of $\Phi$ implies $C_\Phi\geq 0$. Diagonalizing $C_\Phi$ we obtain orthonormal unit vectors $\phi_k$ in $\Hi_{\rho_0}$ and positive numbers $\lambda_k$ such that
$$
C_\Phi=\sum^n_{k=1}\lambda_k|\phi_k\ket\bra\phi_k|.
$$
By definition of $C_\Phi$ it is clear that we can take the $\phi_k$'s from the dense subspace $\pi_{\rho_0}(\Bi(\Hi_0))\Omega$, so we can write $|\phi_k\ket\bra\phi_k|=S_kp_0S_k^*$ for some $S_k\in\pi_{\rho_0}(\Bi(\Hi_0))$. Letting $K_k$ be an operator on $\Hi_0$ such that 
$$
\pi_{\rho_0}(K_k)=\sqrt{\lambda_k}S_k,
$$
we obtain Kraus operators for $\Phi$ which are orthogonal for $\bra\cdot|\cdot\ket_{\rho_0}$. Thus $Q$ is diagonal for these Kraus operators. Moreover, since the $\phi_k$'s are orthonormal we get $\lambda_k=Q_{k,k}/\Tr(Q)$, and in this way we recover the fact \cite[Prop. 10.4]{AlFa1} that the correlation matrix of $\rho_0$ has the same spectrum as $C_\Phi$. 

By linear independence of the $K_k$'s, at most one of them ($K_1$, say) can be proportional to the identity operator $\bone$. Therefore
\begin{equation}\label{zeroonKraus}
\Tr(\rho_0K_j)=\sqrt{\lambda_j}\bra\Omega|S_j\Omega\ket=\sqrt{\lambda_j}\bra\Omega|\phi_j\ket=0
\end{equation}
for all $j=2,\dots,n$ (and possibly also for $j=1$). 

\end{Remark}
Observe that Lemma \ref{LemmaChoidiago} cannot be applied to the products $K_\mathbf{j}$ which give a Kraus representations for $\Phi^m$ because this Kraus representation is already fixed from the choice of $K_1,\dots,K_n$. As mentioned already in Remark \ref{nonminimalremark}, the $K_\mathbf{j}$'s are far from the minimal Kraus representation of $\Phi^m$. Thus, the matrix $Q_m$ cannot be chosen diagonal in general, and there will be nontrivial correlations between outcomes $\mathbf{j}$ and $\mathbf{k}$. 
Indeed, the measurement defined by $\Phi^m$ has only $n_m:=\dim(\GH_m)$ outcomes, and so there will be only $n_m$ post-measurement states in an optimal description of the measurement. 

From now on, $K_1,\dots,K_n$ will always be the $\rho_0$-orthogonal choice of minimal Kraus operators for $\Phi$ guaranteed by Lemma \ref{LemmaChoidiago}. 
\begin{Remark}\label{subofGNSrem} We can identify $K_\mathbf{k}^*K_\mathbf{j}$ with the matrix unit $|p_me_\mathbf{k}\ket\bra p_me_\mathbf{j}|$ in $\Bi(\GH_m)$. Then $\Tr(\rho_0\cdot)$ defines a state $\phi_m$ on $\Bi(\GH_m)$ by the formula
$$
\phi_m(A):=\frac{\Tr(Q_mA)}{\Tr(Q_m)},\qquad \forall A\in\Bi(\GH_m).
$$
In fact, by Remark \ref{GNSremark} we can identify $\GH_m$ with a subspace of $\Hi_{\rho_0}$.  If $\rho_0$ has $\Phi$-symmetric correlations then the subspaces $\GH_m$ of $\Hi_{\rho_0}$ satisfy $\GH_{m+l}\subseteq\GH_m\otimes\GH_l$ for all $m\leq l$ (see \cite[§4.7]{An6}). 
\end{Remark}

\subsection{KMS properties with respect to a channel}
Recall that a state $\omega$ on a $C^*$-algebra $\Ai$ is usually regarded as an ``equilibrium state" if it is \textbf{KMS state} \cite[§5.3]{BrRo2}, i.e. there exists a strongly continuous one-parameter group $\sigma^\omega_\bullet$ of $*$-automorphisms of $\Ai$ such that the \textbf{KMS condition} holds: for all $a,b$ in a dense $*$-subalgebra of $\Ai$ we have
$$
\omega(ab)=\omega(\sigma_i^\omega(b)a),
$$  
where $\sigma_i^\omega$ is a (non-$*$) homomorphism of $\Ai$ obtained by analytically continuing $\R\ni t\to\sigma_t^\omega$ to the imaginary unit $i\in\C$. We refer to $\sigma_\bullet^\omega$ as the \textbf{modular automorphism group} of $\omega$. If $\Ai=\Mn_l(\C)$ is a finite-dimensional matrix algebra, any density matrix $\rho$ of full rank (so, defining a faithful state) is a KMS state, with modular automorphisms
$$
\sigma^\rho_t(A):=\rho^{it}A\rho^{-it},\qquad\forall A\in\Mn_l(\C),
$$
and we have $\sigma_t^\rho(A)=e^{itH}Ae^{-itH}$ if we write $\rho=e^{H}/\Tr(e^{H})$ for some positive matrix $H$. This is seen from
$$
\Tr(\rho AB)=\Tr(\rho \rho^{-1}B\rho A).
$$
We will now see how the correlation matrix \eqref{corrmatrix} can be used to define thermodynamics of the state $\rho_0$ \emph{with respect to} the evolution $\Phi$, in case we assume the ``$Q$-sphere condition" (defined below) on the Kraus operators. 

\begin{dfn}\label{Krausalgnot}
The \textbf{Kraus algebra} is the $C^*$-algebra $\GT$ generated by the Kraus operators $K_1,\dots,K_n$. Let $\GT^{(0)}$ be the $C^*$-subalgebra of $\GT$ generated by elements of the form $K_\mathbf{j}^*K_\mathbf{k}$ and $K_\mathbf{j}K_\mathbf{k}^*$ with $|\mathbf{j}|=|\mathbf{k}|$.
\end{dfn}

\begin{prop}\label{propKMS}
Let $Q\in\Bi(\GH)$ be a positive invertible matrix such that $Q_m:=p_mQ^{\otimes m}p_m$ is equal to $Q^{\otimes m}p_m$ for all $m\in\N_0$, where $p_m:\GH^{\otimes m}\to\GH_m$ is the projection. Suppose that the Kraus operators satisfy the \textbf{$Q$-sphere condition}
\begin{equation}\label{longQsphere}
\sum_{|\mathbf{j}|=m=|\mathbf{k}|}Q^{\mathbf{k},\mathbf{j}}K_\mathbf{j}K_\mathbf{k}^*=\bone
\end{equation}
(for all $m\in\N$). Then there is a KMS state $\omega_Q$ on the Kraus algebra $\GT$ with modular automorphism group $\sigma_\bullet$ given by
\begin{equation}\label{modauto}
\sigma_t(K_\mathbf{j}):=\sum_{|\mathbf{r}|=m}Q^{-it}_{\mathbf{j},\mathbf{r}}K_\mathbf{r},\qquad\textnormal{ when }|\mathbf{j}|=m
\end{equation}
such that (for $|\mathbf{j}|=m$, $|\mathbf{k}|=l$)
$$
\omega_Q(K_\mathbf{j}^*K_\mathbf{k})=\delta_{m,l}\frac{Q_{\mathbf{k},\mathbf{j}}}{\Tr(Q_m)},\qquad \omega_Q(K_\mathbf{j}K_\mathbf{k}^*)=\delta_{m,l}\frac{\bra e_{\mathbf{j}}|p_me_\mathbf{k}\ket}{\Tr(Q_m)}.
$$
\end{prop}
\begin{proof} 
The relation \eqref{longQsphere} ensures that the proposed formulas for $\omega_Q$ define a state on $\GT$ (cf. also the proof of Corollary \ref{normordKMScor} below). We need to check the formula \eqref{modauto} on products involving equally many $K_j$'s as $K_k^*$'s. Assuming \eqref{modauto} we get, for $|\mathbf{j}|=m=|\mathbf{k}|$,
\begin{align*}
\Tr(\rho_0K_\mathbf{j}K_\mathbf{k}^*)&=\Tr(\rho_0K_\mathbf{k}^*\sigma_{-i}(K_\mathbf{j}))=\sum_{|\mathbf{r}|=m}Q^{\mathbf{j},\mathbf{r}}\Tr(\rho_0K_\mathbf{k}^*K_\mathbf{r})
\\&=\frac{1}{\Tr(Q_m)}\sum_{|\mathbf{r}|=m}Q^{\mathbf{j},\mathbf{r}}Q_{\mathbf{r},\mathbf{k}}
\\&=\frac{(p_m)_{\mathbf{j},\mathbf{k}}}{\Tr(Q_m)}.
\end{align*}
We see that formula \eqref{modauto} defines a homomorphism if we use the assumption that $Q^{\otimes m}$ preserves the subspace $\GH_m$. Namely, the fact that $\GH_m$ is the span of the Kraus operators gives
$$
K_\mathbf{k}=\sum_{|\mathbf{r}|=m}(p_m)_{\mathbf{k},\mathbf{r}}K_\mathbf{r},\qquad\textnormal{ when }|\mathbf{k}|=m,
$$
from which one deduces that $\sigma_t$ is a homomorphism. By definition of $\GT$ it follows that $\omega_Q$ is KMS. 
\end{proof}
\begin{dfn}\label{normorddef}
We say that an element $K$ of $\GT$ is \textbf{normally ordered} if $K$ is expressed as a sum of elements of the form $K_\mathbf{j}^*K_\mathbf{k}$.
\end{dfn}
\begin{cor}\label{normordKMScor} Let $\rho_0$ be a density matrix on $\Hi_0$ with $\Phi$-symmetric correlations in the sense of Definition \ref{assumptonrhonoll}, and let $Q_m$ denote its correlation matrix \eqref{corrmatrix}. If every element of the Kraus algebra $\GT$ can be normally ordered then
$$
\Tr(\rho_0K)=\omega_Q(K),\qquad \forall K\in\GT.
$$
In particular, we have the anti-normally ordered correlations (for $|\mathbf{j}|=m=|\mathbf{k}|$)
\begin{equation}\label{antinormcorr}
\Tr(\rho_0K_\mathbf{j}K_\mathbf{k}^*)=\frac{\bra e_{\mathbf{j}}|p_me_\mathbf{k}\ket}{\Tr(Q_m)}.
\end{equation}
\end{cor}
\begin{proof} Applying the state $\Tr(\rho_0\cdot)$ on both sides of \eqref{longQsphere} gives
$$
\sum_{|\mathbf{j}|=m=|\mathbf{k}|}Q^{\mathbf{k},\mathbf{j}}\Tr(\rho_0K_\mathbf{j}K_\mathbf{k}^*)=1,
$$
which is solved by \eqref{antinormcorr}, where we note that
\begin{align*}
&\sum_{|\mathbf{j}|=m=|\mathbf{k}|}Q^{\mathbf{k},\mathbf{j}}\bra e_{\mathbf{j}}|p_me_\mathbf{k}\ket=\sum_{|\mathbf{k}|=m}(p_mQ_m^{-1})_{\mathbf{k},\mathbf{k}}
\\&=\Tr(p_mQ_m^{-1})=\Tr(Q_m^{-1})=\Tr(Q_m).
\end{align*}
However, there is the ambiguity that replacing $\bra e_{\mathbf{j}}|p_me_\mathbf{k}\ket$ by $\bra e_{\mathbf{k}}|p_me_\mathbf{j}\ket=\overline{\bra e_{\mathbf{j}}|p_me_\mathbf{k}\ket}$ would give the same result, since $\Tr(Q_m^t)=\Tr(Q_m)$, where $Q^t_m$ is the transpose of $Q_m$. Similarly, $\Tr(\rho_0K_\mathbf{j}K_\mathbf{k}^*)=Q_{\mathbf{j},\mathbf{k}}/\Tr(p_m)$ is consistent with the $Q$-sphere condition. On the other hand, since we assume that it is possible to use commutation relations to relate $K_\mathbf{j}K_\mathbf{k}^*$ and $K_\mathbf{k}^*K_\mathbf{j}$, only one of these choices can give a well-defined state on $\GT$. Therefore, Proposition \ref{propKMS} shows that equation \eqref{antinormcorr} is the correct thing. 

Finally, in the same way as \eqref{zeroonKraus} we deduce that
$$
\Tr(\rho_0K_\mathbf{j}K_\mathbf{k}^*)=0=\Tr(\rho_0K_\mathbf{k}^*K_\mathbf{j}),\qquad \text{when }|\mathbf{j}|\ne|\mathbf{k}|,
$$
unless $K_\mathbf{j}K_\mathbf{k}^*$ can be rewritten in the form $K_\mathbf{r}K_\mathbf{s}^*$ with $|\mathbf{r}|=|\mathbf{s}|$, and similarly for $K_\mathbf{k}^*K_\mathbf{j}$.
\end{proof}

\begin{Remark}\label{normordKMSrem}
If normal ordering is not possible in $\GT$ then, even if the $Q$-sphere condition holds, we cannot exclude the possibility that
$$
\Tr(\rho_0K_\mathbf{j}K_\mathbf{k}^*)=\frac{Q_{\mathbf{j},\mathbf{k}}}{\Tr(p_m)},
$$
and in this case $\Tr(\rho_0\cdot)$ does not restrict to a KMS state on $\GT$. The $Q$-sphere condition always guarantees that there is a KMS state on $\GT$ with the same normally ordered correlations as $\Tr(\rho_0\cdot)$, but we cannot always assure that this KMS state coincides with $\Tr(\rho_0\cdot)$.
\end{Remark}

\section{Detailed balance}

\subsection{Motivation}
Let us first recall an existing notion of detailed balance for quantum channels \cite{Crook1}, motivated as follows. Consider a classical Markov process for which $\{1,\dots,n\}$ is the set of possible states and $M\in\Mn_n(\C)$ is the transition matrix, meaning that $M_{j,k}$ is the probability of transition from state $k$ to state $j$. If the system is in an equilibrium distribution then a time reversal should have no effect on that distribution, and the probability of observing the transition $j\to k$ in the forward chain should be the same as the probability of observing the transition $k\to j$ in the time-reversed chain. Because the equilibrium probability distribution $\pi=(\pi_1,\dots,\pi_n)\in\R^n$ is the same for both chains, the transition matrix $\hat{M}(t)$ of the backward process must satisfy 
$$
M(t)_{j,k}\pi(t)_k=\hat{M}(t)_{k,j}\pi(t)_j,\qquad \forall j,k\in\{1,\dots,n\}.
$$
In matrix notation this says
\begin{equation}\label{ClassTR}
\hat{M}=\diag(\pi)^{-1}M^T\diag(\pi),
\end{equation}
where $\diag(\pi)$ is the diagonal matrix with entries $\pi_1,\dots,\pi_n$. The matrix $\hat{M}$ is the \textbf{time reversal} (or $\pi$-dual) of $M$. The Markov process is said to satisfy \textbf{detailed balance} if $\hat{M}=M$. That is, detailed balance holds iff
\begin{equation}\label{classdetbalance}
M(t)_{j,k}\pi(t)_k=M(t)_{k,j}\pi(t)_j
\end{equation}
for all $j,k\in\{1,\dots,n\}$. 

In view of the above, Crooks introduced the following notion of time reversal of a quantum operation $\Phi_*:\Bi(\Hi_0)_*\to\Bi(\Hi_0)_*$ \cite[§III]{Crook1}. Suppose that $\Phi_*$ has a unique fixed point, i.e. there is a unique density matrix $\rho_0$ such that $\Phi_*(\rho_0)=\rho_0$. Then $\rho_0$ plays the role of equilibrium distribution. We refer to $\Phi_*$ as the ``forward process" and we want to find a sensible time reversal of $\Phi_*$ which constitutes the ``backward process".

For each Kraus operator $K_j$ of $\Phi_*$ there should be a corresponding operator $\bar{K}_j$ of the reversed process such that, starting from the state $\rho_0$, the probability of observing any sequence of Kraus operators in the forward dynamics is the same as the probability of observing the reversed sequence of reversed operators in the reversed dynamics \cite[§III]{Crook1}. For instance, for all $j,k=1,\dots,n$ we should have
\begin{equation}\label{CrooksargTR}
\Tr(K_kK_j\rho_0 K_j^*K_k^*)=\Tr(\bar{K}_j\bar{K}_k\rho_0 \bar{K}_k^*\bar{K}_j^*).
\end{equation}
The invariant density matrix $\rho_0$ is assumed to be invertible. Then the solution to \eqref{CrooksargTR} is \cite[§III]{Crook1}
\begin{equation}\label{CrooksKraussol}
\bar{K}_j=\rho_0^{1/2}K_j^*\rho_0^{-1/2},\qquad \forall j=1,\dots,n.
\end{equation}
The \textbf{reversal} (or $\rho_0$-dual) of $\Phi_*$ is then defined as the channel $\bar{\Phi}_*$ with Kraus operators $\bar{K}_j$,
$$
\bar{\Phi}_*(\rho):=\sum^n_{j=1}\bar{K}_j\rho \bar{K}_j^*,\qquad \forall \rho\in\Bi(\Hi_0)_*.
$$

We would like to construct a unitary model for the time reversal $\bar{\Phi}$ in terms of the unitary model for $\Phi$ furnished by the Stinepring representation. The reason for this is that if the Stinespring representations can say something about the structure of $\Phi$, comparison with the Stinespring representation of $\bar{\Phi}$ would perhaps tell us something about the time-irreversibility of the dynamics.

The time reversal of the reduced dynamics obtained from a system-bath interaction was obtained approximately in \cite[§VII]{Crook1} for weak couplings and equilibrium bath states. However, the dynamics so obtained is not trace-preserving in general, because the conjugate transpose matrix $W^c:=(W^t)^*$ is not always unitary. Thus, even if we restrict ourselves to weak couplings and accept an approximation, we do not get a quantum channel for the backward process if its is obtained by tracing out the same bath. In order to have a unitary model for the channel with Kraus operators \eqref{CrooksKraussol} which relates to the unitary model of $\Phi$, something extra is needed.

\subsection{Putting it all together}
To summarize, the time-reversal of $(\Phi^m)_{m\in\N_0}$ should have the opposite ordering of the Kraus operators. Since $\Phi$ is a sum of terms $K_j^*\cdot K_j$, we therefore have to swap $K_j$ with its adjoint $K_j^*$. Simply switching $W$ and $W^*$ will give something entirely different, as one easily checks. Instead, we are led to the transposed adjoint $W^c:=(W^*)^t$ in searching for a dilation model for $\bar{\Phi}$. However, since the elements $W_{j,k}\in\Bi(\Hi_0)$ of the $n\times n$ matrix $W$ do not commute in general, it rarely happens that $W^c$ is unitary. 

Suppose generally that $F$ is an invertible $n\times n$ matrix such that 
$$
\bar{W}:=(\bone\otimes F)W^c(\bone\otimes F^{-1})
$$
is unitary on $\Hi_0\otimes\overline{\GH}$, where $\overline{\GH}$ is the conjugate Hilbert space of $\GH$. We then define a \textbf{time-reversal} of $\Phi$ as
\begin{align}\label{TRgenF}
\bar{\Phi}^F(A)&:=(\id\otimes\Tr)\big((\bone\otimes|\bar{e}_1\ket\bra \bar{e}_1|)\bar{W}^{-1}(A\otimes\bone)\bar{W}\big),\qquad\forall A\in\Bi(\Hi_0).
\end{align}
The operation $\bar{\Phi}$ is a channel because $\bar{W}$ is assumed to be unitary. The Kraus operators $\bar{K}_j$ of $\bar{\Phi}^F$ are given by
\begin{equation}\label{KrausTR}
\bar{K}_k=(FW^cF^{-1})_{1,k}=\sum^n_{r,s=1}F_{1,r}W_{r,s}^*F_{s,k}^{-1},\qquad \forall k=1,\dots,n.
\end{equation}
\begin{dfn}\label{TRinvdef}
The quantum channel $\Phi$ is said to be \textbf{time-reversal invariant} if there is an invertible $F\in\Bi(\GH)$ for which 
$$
\bar{\Phi}^F=\Phi.
$$
\end{dfn}
The motivation for Definition \ref{TRinvdef} will be given in §\ref{interpretsec}. In this section we will focus on the case when $F=Q^{1/2}$ for a correlation matrix $Q$. To arrive at our notion of detailed balance we need a lemma. 

\begin{Lemma}\label{LemmageneralTR}
 Let $\GT$ be the $C^*$-algebra generated by the Kraus operators $K_1,\dots,K_n$ and suppose that $\omega$ is a KMS state on $\GT$. Define elements $\bar{K}_1,\dots,\bar{K}_n$ of $\GT$ by requiring that
$$
\omega(\bar{K}_{\bar{\mathbf{k}}}^*\bar{K}_{\bar{\mathbf{k}}})=\omega(K_\mathbf{k}^*K_\mathbf{k})
$$
for all multi-indices $\mathbf{k}=k_1\cdots k_m$, where $\bar{\mathbf{k}}:=k_m\cdots k_1$. Then 
\begin{equation}\label{TRKrausKMS}
\bar{K}_j=\sigma_{-i/2}(K_j^*),\qquad\forall j=1,\dots,n,
\end{equation}
where $\sigma_\bullet$ is the modular automorphism group of $\omega$. 
\end{Lemma}
\begin{proof} Note first that $\sigma_{i\lambda}(T)^*=\sigma_{-i\lambda}(T^*)$ for all $T\in\GT$ and all $\lambda\in\R$, so if we take $\bar{K}_j$ to be given by \eqref{TRKrausKMS} then its adjoint is
$$
\bar{K}_j^*=\sigma_{i/2}(K_j).
$$
More generally,
\begin{align*}
\bar{K}_\mathbf{j}^*&=\big(\sigma_{-i/2}(K_{j_1}^*)\cdots\sigma_{-i/2}(K_{j_m}^*)\big)^*
=\sigma_{i/2}(K_{j_m})\cdots\sigma_{i/2}(K_{j_1})=\sigma_{i/2}(K_{\bar{\mathbf{j}}}),
\end{align*}
where we note the interchange $\mathbf{j}\leftrightarrow\bar{\mathbf{j}}$. Now the KMS condition says
\begin{align*}
\omega\big(\sigma_{i/2}(K_\mathbf{k})\sigma_{-i/2}(K_\mathbf{k}^*)\big)&=\omega(K_\mathbf{k}^*K_\mathbf{k}),
\end{align*}
so the result is clear. 
\end{proof}

Let now $Q$ be the correlation matrix of a density matrix $\rho_0$ on $\Hi_0$ such that $K_k\rho_0K_k^*\ne 0$ for all $k=1,\dots,n$ (recall Definition \ref{corrmatrixdef}). The discussion about time-reversal symmetry would be greatly simplified if it happens that the operator 
\begin{equation}\label{conjugateunitprocess}
\bar{W}:=(\bone\otimes Q^{1/2})W^c(\bone\otimes Q^{-1/2})
\end{equation}
is unitary, where we have chosen the positive square root $Q^{1/2}$ of $Q$. Namely, we have seen that we can take $Q$ to be diagonal and if we set $Q_{1,1}=1$ then the Kraus operators $\bar{K}_k$ of $\bar{\Phi}$ are simply given by
\begin{equation}\label{KrausTR}
\bar{K}_k=(Q^{1/2}W^cQ^{-1/2})_{1,k}=Q_{k,k}^{-1/2}K_k^*,\qquad \forall k=1,\dots,n.
\end{equation}
Now the point is that \eqref{conjugateunitprocess} is unitary whenever $\Phi$ is symmetric under $\G\subseteq A_u(Q)$ (cf. \eqref{sqrootconjugate} below), or more generally whenever the $Q$-sphere condition \eqref{choiceminimalrels} holds. In fact, \eqref{conjugateunitprocess} defines a unitary operator \emph{precisely} when the $Q$-sphere condition holds. 

Moreover, we obtain Crooks' condition \eqref{CrooksargTR} in this formalism.

\begin{thm}\label{ThmCrookscond}
Suppose that $\rho_0$ has $\Phi$-symmetric correlations, and that the Kraus operators of $\Phi^m$ satisfy the $Q_m$-sphere condition \eqref{longQsphere} and allow for normal ordering. Define $\bar{K}_k$ by \eqref{KrausTR}. Then for all multi-indices $\mathbf{k}$ of length $|\mathbf{k}|=m$, 
$$
\Tr(\rho_0\bar{K}_{\bar{\mathbf{k}}}^*\bar{K}_{\bar{\mathbf{k}}})=\Tr(\rho_0K_\mathbf{k}^*K_\mathbf{k}),
$$
where $\bar{\mathbf{k}}:=k_m\cdots k_1$ for $\mathbf{k}=k_1\cdots k_m$. 
\end{thm}
\begin{proof} Observe that, with $\sigma_\bullet$ defined by \eqref{modauto}, we have
$$
\bar{K}_k=\sigma_{-i/2}(K_k^*),\qquad\forall k=1,\dots,n,
$$
and hence also $\bar{K}_{\bar{\mathbf{k}}}=\sigma_{-i/2}(K_\mathbf{k}^*)$ for all $\mathbf{k}$, so the result follows by reversing the proof of Lemma \ref{LemmageneralTR}. 
\end{proof}
Note that there is no need for $\rho_0$ to have dense range here. 

We can summarize our discussion in a suitable definition of detailed balance.
\begin{dfn}\label{detbaldef}
Let $Q\in\Bi(\GH)$ be a positive invertible operator such that $Q_m=p_mQ^{\otimes m}p_m=Q^{\otimes m}p_m$ for all $m\in\N_0$. Let $\rho_0$ be a density matrix on $\Hi_0$ such that $\Tr(\rho_0\cdot):\Bi(\Hi_0)\to\C$ extends the KMS state $\omega_Q:\GT\to\C$ defined in Proposition \ref{propKMS}. The dynamics $(\Phi^m)_{m\in\N_0}$ satisfies \textbf{detailed balance} with respect to $\rho_0$ if the $\rho_0$-orthogonal Kraus operators of $\Phi$ (see Lemma \ref{LemmaChoidiago}) satisfy the $Q$-sphere condition \eqref{longQsphere}.
\end{dfn}
In particular, suppose that normal ordering is possible in the Kraus algebra $\GT$ and that $\rho_0$ is a density matrix on $\Hi_0$ with $\Phi$-symmetric correlations equal to $(Q_m)_{m\in\N_0}$. Then Corollary \ref{normordKMScor} shows that detailed balance holds with respect to $\rho_0$ if the $Q$-sphere condition is satisfied by the $\rho_0$-orthogonal Kraus operators of $\Phi$.

Definition \ref{detbaldef} is more complicated because we want to allow for detailed balance also when there is no normal ordering. Nonetheless, the absence of normal ordering corresponds to a rather extreme form of complicated dynamics; see §\ref{interpretsec}.

\begin{Remark}[Fluctuation relations]
If the time evolution satisfies detailed balance in the sense of Definition \ref{detbaldef} then the path probabilities in any two-point measurement satisfy a kind of ``microscopic reversibility" (cf. \cite{Monn1}), from which nonequilibrium fluctuation relations follow immediately (for any number $m\in\N$ of time steps). We shall not discuss fluctuation relations here but we realize that there is an overlap between our discussion and a very recent paper \cite{MHP1}. The time-reversal (anti-) unitary appearing there (which may be referred to as the ``kinematical time reversal") could easily be included also here. Then \cite{MHP1} gives examples of our formalism with $\Hi_0$ finite-dimensional and $\rho_0$ of full rank.  
\end{Remark}

\section{Interpretation and quantum groups}

\subsection{Compact matrix quantum groups}
We give a brief but self-contained introduction to compact matrix quantum groups \cite{KlS}, \cite{Timm1}, \cite{Wor1} which will help in identifying symmetries of quantum systems. 

The classical compact group $G:=\SU(2)$ can be defined in terms of the $C^*$-algebra $C(G)$ of continuous function on $G$ with values in $\C$. Namely, $C(G)$ is generated by the coordinate functions $u_{j,k}:G\to\C$ which sends a matrix $g=\big(\begin{smallmatrix}g_{1,1}&g_{1,2}\\g_{2,1}&g_{2,2}\end{smallmatrix}\big)$ to its coefficients $g_{j,k}$,
$$
u_{j,k}(g):=g_{j,k},\qquad j,k=1,2.
$$
Thus, $C(G)$ is characterized as the $C^*$-algebra generated by the entries $u_{j,k}$ of a $2\times 2$ unitary matrix 
$$
u=(u_{j,k})_{j,k=1,2}\in\Mn_2(\C)\otimes C(G)
$$
such that the entries $u_{j,k}$ mutually commute, or equivalently such that 
\begin{equation}\label{BuQunitdef}
u=(F\otimes\bone)u^c(F^{-1}\otimes\bone)
\end{equation}
where $F$ is the invertible matrix $F:=\big(\begin{smallmatrix}0&1\\-1&0\end{smallmatrix}\big)$ and $u^c:=(u^t)^*$ is the matrix with entries $u_{j,k}^*$ (i.e. $u^c$ is the adjoint of the transpose of $u$). The group multiplication in $G$ is encoded in the ``comultiplication" $\Delta:C(G)\to C(G)\otimes C(G)$ given by
\begin{equation}\label{comult}
\Delta(u_{j,k}):=\sum_ru_{j,r}\otimes u_{r,k}
\end{equation}
(this follows from how matrix multiplication is defined), the inversion $G\ni g\to g^{-1}$ is equivalently described by the ``antipode" $\kappa$ defined as
\begin{equation}\label{antipode}
\kappa(u_{j,k}):=u_{k,j}^*
\end{equation}
and the unit $1\in G$ is captured by the ``counit" $\ep(u_{j,k}):=\delta_{j,k}$ (where $\delta_{j,k}$ is the Kronecker delta).

These observation lay the foundations of the theory of compact quantum groups: one looks for $C^*$-algebras possessing maps with the same properties as the maps $\Delta$, $\ep$, $\kappa$. There is no underlying group when the $C^*$-algebra is noncommutative, but the analogue is so good that one speaks of ``quantum groups". Admittedly this terminology is misleading but facilitates the discussion about ``representations" of the quantum group by making analogues with compact groups. We shall here only be interested in compact ``matrix" quantum groups where one of the representations is given a distinguished role.

\begin{dfn}[{\cite{Wang3}, \cite{Wor3}}] 
A \textbf{compact matrix quantum group} $\G$ is defined by a unital $C^*$-algebra $C(\G)$ generated by the entries $u_{j,k}$ of a single unitary matrix $u\in\Mn_n(\C)\otimes C(\G)$ (for some $n\in\N$) such that
\begin{enumerate}[(i)]
\item{the map $\Delta:C(\G)\to C(\G)\otimes C(\G)$ defined by the formula \eqref{comult} satisfies $\Delta(fg)=\Delta(f)\Delta(g)$ and $\Delta(f^*)=\Delta(f)^*$ for all $f,g\in C(\G)$, and}
\item{the matrix transpose $u^t$ of $u$ is invertible.}
\end{enumerate}
\end{dfn} 
We refer to the generating matrix $u$ as the \textbf{defining representation} of the ``group" $\G$. 

\begin{Remark}\label{noaxiomcoass}
If $\Delta$, when defined as in \eqref{comult}, is a homomorphism then $\Delta$ automatically satisfies ``coassociativity" and there is no need to postulate this property in the definition of a compact matrix quantum group. In fact, \eqref{comult} is the unique homomorphism on $C(\G)$ satisfying coassociativity since $C(\G)$ is generated by the $u_{j,k}$'s. The desired properties of the map $\kappa(u_{j,k}):=u_{k,j}^*$ follows from the invertibility of $u$ and its transpose $u^t$ \cite{Wor3}.
\end{Remark}

In the following, for a matrix $u$ with entries in $C(\G)$, we again write $u^c$ for the transpose of $u^*$, i.e. $(u^c)_{j,k}:=u_{j,k}^*$ where $u_{j,k}^*$ is the adjoint of $u_{j,k}$ in $C(\G)$.
\begin{dfn}[{\cite{Wang3}, \cite[Déf. 1]{Ban4}}]\label{UCMQGdef}
Let $F\in\GeL(n,C)$ be an invertible matrix and write $Q:=F^*F$. The \textbf{universal unitary quantum group} is the compact matrix quantum group $\G=A_u(Q)$ whose algebra of continuous functions $C(\G)$ is generated by the entries of a unitary $n\times n$ matrix $u$ satisfying the relations making $(F\otimes\bone)u^c(F^{-1}\otimes\bone)$ a unitary matrix. 

The \textbf{universal orthogonal quantum group} $\G=B_u(F)$ is the compact matrix quantum group whose algebra $C(\G)$ is the quotient of that of $A_u(Q)$ by the relation \eqref{BuQunitdef}. 
\end{dfn}
We have thus seen an example of a $B_u(F)$ already: the algebra $C(G)$ for the classical group $G:=\SU(2)$. More generally, the quantum special unitary group $\SU_q(2)$ is of the form $B_u(F)$, for $q\in[-1,1]$. 

We haven't defined the notion of ``representation" of a quantum group but it suffices to say that it is given by an element $v\in\Bi(\GH_v)\otimes C(\G)$ just as for the defining representation $u\in\Mn_n(\C)\otimes C(\G)$ but where $\GH_v$ is a Hilbert space more general than $\C^n$ and $v$ is allowed to be merely invertible (not necessarily unitary). Representations will here be understood to be finite-dimensional, i.e. each $\GH$ is finite-dimensional.
\begin{Remark}[Conjugates]\label{conjugateremark}
Recall that the ``conjugate" of a representation $v$ of a classical compact group on a Hilbert space $\GH_v$ is the representation $v^c$ on the conjugate Hilbert space $\bar{\GH}_v=\GH^*_v$. The crux with the quantum case is that $v^c$ is not a unitary matrix in general; this is so because the entries of $v$ need not commute. Nevertheless, for any representation $v$ of a quantum group $\G$, one can find an invertible matrix $F$ such that 
\begin{equation}\label{conjugatedef}
\bar{v}:=(F\otimes\bone)v^c(F^{-1}\otimes\bone)
\end{equation}
is unitary, and we can take this as a definition of a \textbf{conjugate} of $v$. The definition of $A_u(Q)$ ensures that this minimal requirement is satisfied for the defining representation, and this is what makes $A_u(Q)$ universal (see below). The relation \eqref{BuQunitdef} says precisely that $\bar{u}=u$, i.e. that $u$ is \textbf{self-conjugate} for every $B_u(F)$. 
\end{Remark}
Suppose that $\Hb$ and $\G$ are compact matrix quantum groups such that $C(\Hb)$ is a quotient of $C(\G)$. If the quotient map $\pi:C(\G)\to C(\Hb)$ fulfills $(\pi\otimes\pi)\circ\Delta_\G=\Delta_{\Hb}\circ\pi$, i.e. if $\pi$ intertwines the comultiplication of $\G$ with that of $\Hb$, then $\Hb$ is a \textbf{quantum subgroup} $\G$. As an example, $B_u(F)$ is a quantum subgroup of $A_u(Q)$ when $F^*F=Q$. 

The name ``universal" is motivated by the following fact, which we should anticipate from \eqref{conjugatedef}. 
\begin{Lemma}[{\cite{VaDW}}] Any compact matrix quantum group $\G$ is a quantum subgroup of $A_u(Q)$ for some $Q$. If $\G$ in addition has a self-conjugate defining representation, then $C(\G)$ is a quantum subgroup of some $B_u(F)$. We write $\G\subset A_u(Q)$ and $\G\subset B_u(F)\subset A_u(Q)$ for these cases respectively.
\end{Lemma}

\begin{Remark}[Choice of conjugate]\label{unitarizationremark}
The choice of conjugate \eqref{conjugatedef} is not unique but the unitary equivalence class of $\bar{u}$ does not depend on the choice $F$; the requirement is that $\bar{u}$ is \emph{unitary}. For $B_u(Q)$ we could clearly choose the matrix $F$ appearing in the definition $Q=F^*F$ but there is another possibility. For any representation $v$ of a compact matrix quantum group $\G$, there is a positive invertible matrix $Q_v$ such that  
\begin{equation}\label{sqrootconjugate}
\bar{v}:=(Q_v^{1/2}\otimes\bone)v^c(Q_v^{-1/2}\otimes\bone)
\end{equation}
is unitary. When $v=u$ is the defining representation of $\G$, the matrix $Q_v=Q$ coincides with the one for which $C(\G)$ is a quotient of $A_u(Q)$. For any scalar $\lambda>0$, the matrix $\lambda Q$ does the same job and hence $A_u(\lambda Q)\cong A_u(Q)$.  
\end{Remark}

\begin{Remark}[Terminology] The name ``orthogonal" quantum group attached to $B_u(F)$ is motivated by the fact that $F=\bone\in\GeL(n,\C)$ gives the so-called ``free orthogonal group" $\On^+(n)$ which is obtained from the classical orthogonal group by just removing commutativity of the $u_{j,k}$'s. However, the \emph{classical} orthogonal groups $\On(n)=\On(n,\R)$ cannot be obtained as a $B_u(F)$ since choosing the $F$ which gives commutativity does not impose selfadjointness of the generators $u_{j,k}$. Hence ``universal quantum groups with self-conjugate defining representation" would be a more appropriate name for $B_u(F)$. On the other hand, the classical \emph{unitary} groups $\Un(n)$ (as well as their free analogues) can indeed be obtained from the \emph{unitary} quantum groups $A_u(Q)$.
\end{Remark}

As in the classical case, there is a notion of ``irreducible" representations of quantum groups. For such a representation $v\in\Bi(\GH_v)\otimes C(\G)$, the action on $\Bi(\GH_v)$ given by
$$
\Bi(\GH_v)\ni A\to v(A\otimes\bone)v^* 
$$
leaves only the scalar multiplies of the identity operator ``invariant", i.e. $v(A\otimes \bone)v^*=A\otimes\bone$ implies $A\in\C\bone$. Moreover, there is only one invariant state on $\Bi(\GH_v)$ for this action. If $Q_v$ is the positive invertible matrix such that \eqref{sqrootconjugate} is unitary and $\Tr(Q_v)=\Tr(Q^{-1}_v)$ then the unique invariant state is given by the density matrix
\begin{equation}\label{invstate}
\rho_v:=\frac{Q_v^t}{\Tr(Q_v)}.
\end{equation}
For a classical compact group, $\rho_v=\bone/\dim(\GH_v)$ is the only possibility.  

When discussing a compact matrix quantum group, we shall assume that the defining representation $u$ is irreducible. For $B_u(F)$ we have $u=(FF^c)u(FF^c)^{-1}$, i.e. $FF^c$ intertwines the defining representation $u$ with itself. We have to assume $FF^c$ to be a scalar multiple of the identity matrix in order to have $u$ irreducible. For $A_u(Q)$ it is enough that $Q>0$ for $u$ to be irreducible.
\begin{Remark}[$B_u(F)$ has no classical analogue]\label{GlebschGordanrem}
Recall that $\SU(2)$ has self-conjugate defining representation, whereas this is not true for $\SU(n)$ when $n\geq 3$. Not only does $B_u(F)$ have self-conjugate defining representation $(u,\GH)$, but the tensor products $\GH_m\otimes\GH_l$ decompose into irreducibles in exactly the same way as for $\SU(2)$. One may therefore say that $B_u(F)$ is a ``higher-dimensional quantum $\SU(2)$ group", in general lacking a classical counterpart. 
\end{Remark}

\begin{Remark}[Notation] For comparison with the literature we mention that our notation is that of \cite{VaDW}. Some authors follow \cite[Déf. 1]{Ban4} and write $A_o(F):=B_u(F)$ for the orthogonal group while for the unitary group they write $A_u(F):=A_u(Q)$ where in that latter case $F^*F=Q$. Sometimes also $\On^+_F$ and $\Un^+_F$ are used in place of $A_o(F)$ and $A_u(F)$, and the latter refer then to the $C^*$-algebras $C(\On^+_F)$ and $C(\Un^+_F)$. Finally, there is a notion of ``dual" (discrete) quantum group for any compact quantum group, and the duals of $B_u(F)$ and $A_u(Q)$ are referred to as the ``universal orthogonal free group" and ``universal unitary free group"  respectively, denoted by $\F O(F)$ and $\F U(F)$. 
\end{Remark}

\subsection{The $Q$-sphere condition for quantum groups}
Similar to how we obtained the sequence $(\GH_m)_{m\in\N_0}$ of Stinespring spaces from a quantum channel in §\ref{Stinesection}, the defining representation $(u,\GH)$ of a compact quantum group $\G$ generates a sequence $(\GH_m)_{m\in\N_0}$ of irreducible representations $(u^{(m)},\GH_m)$ of $\G$, satisfying $\GH_{m+l}\subseteq\GH_m\otimes\GH_l$. Namely, we could take $\GH_m$ to be the linear span of products $z_\mathbf{k}$ where $|\mathbf{k}|=m$ and $z_j:=u_{1,j}$ are the elements of the first row of the defining unitary $u$. The representation $u^{(m)}$ of $\G$ on $\GH_m$ is given by
$$
u^{(m)}=(p_m\otimes\bone)u^{\otimes m}(p_m\otimes\bone)\in\Bi(\GH_m)\otimes C(\G),
$$
where $p_m:\GH^{\otimes m}\to\GH_m$ is the projection and $u^{\otimes m}\in\Bi(\GH^{\otimes m})\otimes C(\G)$ is the matrix whose entries in the product basis is given by
$$
u_{\mathbf{j},\mathbf{k}}=u_{j_1,k_1}\cdots u_{j_m,k_m}.
$$
Suppose that $\G\subseteq A_u(Q)$ with $Q$ equal to its transpose $Q^t$ (otherwise replace $Q$ by $Q^t$ in the following). From results described in detail in \cite[§1]{Ban6} one gets that taking $v=u^{\otimes m}$ in \eqref{sqrootconjugate} gives $Q_v=Q^{\otimes m}$. The operator $Q^{\otimes m}$ on $\GH^{\otimes m}$ preserves the subspace $\GH_m$. With $v=u^{(m)}$ in \eqref{sqrootconjugate} we get $Q_v=Q_m^t:=(Q^{\otimes m}p_m)^t$. 

The following is the reason why we use the same notation $Q_m$ here as for the correlation matrix in our formulation of detailed balance.
\begin{prop}\label{propGlongQ} Suppose that $\G\subseteq A_u(Q)$ with $Q_{1,1}=1$ and that a basis $e_1,\dots,e_n$ of $\GH$ in which $Q=Q^t$ is chosen such that $e^{\otimes m}_1\in\GH_m$ for all $m$. Then, for all $m\in\N$, the first-row elements $z_1:=u_{1,1},\dots,z_n:=u_{1,n}$ satisfy the $Q$-sphere condition
\begin{equation}\label{GlongQsphere}
\sum_{|\mathbf{j}|=m=|\mathbf{k}|}Q^{\mathbf{k},\mathbf{j}}z_\mathbf{j}^*z_\mathbf{k}=\bone,
\end{equation}
where $Q^{\mathbf{k},\mathbf{j}}:=\bra e_\mathbf{k}|Q_m^{-1}e_\mathbf{j}\ket$.
\end{prop}
\begin{proof} Equation \eqref{sqrootconjugate} with $v=u^{(m)}$ can be written as (for $m=1$, this is an equivalent way of writing the defining relations for $A_u(Q)$)
\begin{equation}\label{unversalmatrixgrouprel}
u^{(m)t}Q_m^tu^{(m)c}(Q^{-1}_m)^t=\bone=Q_m^tu^{(m)c}(Q^{-1}_m)^tu^{(m)t},
\end{equation}
where $A^t$ denotes the transpose of a matrix $A$, and we omit the $\otimes\bone$ factor. The second equation in \eqref{unversalmatrixgrouprel} can be rewritten as $(Q^{-1}_m)^t=u^{(m)c}(Q^{-1}_m)^tu^{(m)t}$, which yields
$$
Q^{\mathbf{s},\mathbf{r}}=\sum^n_{|\mathbf{j}|=m=|\mathbf{k}|}(u^{(m)}_{\mathbf{r},\mathbf{j}})^*Q^{\mathbf{k},\mathbf{j}}u^{(m)}_{\mathbf{s},\mathbf{k}}.
$$
Setting $\mathbf{r}=\mathbf{s}=1\cdots 1$ we get \eqref{GlongQsphere}. 
\end{proof}

\begin{Remark} We can also deduce \eqref{GlongQsphere} by taking $v=u^{\otimes m}$ in \eqref{sqrootconjugate}, which gives, since $Q^{\otimes m}$ is diagonal,
$$
\bone=\sum_{|\mathbf{k}|=m}z_\mathbf{k}^*(Q^{\otimes m})^{-1}_{\mathbf{k},\mathbf{k}}z_\mathbf{k}.
$$
Using $z_\mathbf{k}^*=\sum_{|\mathbf{j}|=m}(p_m)_{\mathbf{k},\mathbf{j}} z_\mathbf{j}^*$ or $z_\mathbf{j}=\sum_{|\mathbf{j}|=m}(p_m)_{\mathbf{k},\mathbf{j}} z_\mathbf{k}$ together with $Q_m=Q^{\otimes m}p_m$, we again obtain \eqref{GlongQsphere}. 
\end{Remark}

\subsection{Interpretation}\label{interpretsec}
As we have seen, the $Q$-sphere condition discussed in this paper is deeply connected to the theory of compact quantum groups. Applications of our detailed-balance condition will take further advantage of this kind of quantum symmetry, so let us make a few comments on this.

If the evolution described by $\Phi$ is the result of very complicated interactions which appear \emph{completely random} on the relevant time scale (cf. \cite{AcLu1}), then there should be no other relations than the unitality condition \eqref{minimalrels} and the $Q$-sphere condition \eqref{choiceminimalrels} imposed by the Kraus operators. In that case the $n$-tuple $K_1^*,\dots,K_n^*$ satisfy, \emph{by definition}, the same relations as the first row of the defining representation of the universal unitary quantum group $A_u(Q)$, and we say that $\Phi$ has $A_u(Q)$-\textbf{symmetry}.

Free commutation relations holds if the Kraus operators $K_j\in\Bi(\Hi_0)$ are realized by large random matrices. The well-established strategy for using random matrices in quantum physics \cite{Wign1}, \cite{Guhr1}, \cite{Pain1} is to identify the symmetries of the system (such as time-reversal symmetry, Hermiticity of the Hamiltonian etc.) and then restrict the ensemble of random matrices accordingly. If the result fails to reproduce experimental data then one may conclude that some additional symmetry is present which has not been taken into account. Similarly, we here suggest that imposing restrictions of these matrices such that \eqref{minimalrels} holds enables the random dissipative dynamics $\Phi$ to drive the system into nontracial states.

We now go one step further and identify symmetries, the result being that the additional symmetries make the ``ensemble of random matrices" consist of matrices which are in fact not random at all. 
\begin{dfn}\label{Gsymmdef}
We say that a quantum channel $\Phi$ has $\G$-\textbf{symmetry} for some compact matrix quantum group $\G$ if the adjoints $K_1^*,\dots,K_n^*$ of the Kraus operators of $\Phi$ satisfy the same relations as the first row of the defining representation of $\G$.
\end{dfn}
From the definition of the Stinespring spaces $\GH_m$ in §\ref{Stinesection}, it follows that if $\Phi$ has $\G$-symmetry then each $\GH_m$ carries an irreducible representation of $\G$. 

As mentioned, any compact matrix quantum group is a ``quantum subgroup" of $A_u(Q)$ for some $Q\in\GeL(n,\C)$ meaning that, if $\Phi$ has $\G$-symmetry, there will always be a $Q$ such that the $Q$-sphere condition \eqref{longQsphere} holds (due to Proposition \ref{propGlongQ}).

If $\rho_0$ has $\Phi$-symmetric correlations $Q$, the Stinespring space $\GH_m$ can be endowed with the $Q_m$-inner product without spoiling the property $\GH_{m+l}\subset\GH_m\otimes\GH_l$ \cite[§4.7]{An6}. If $\Phi$ has $\G\subset A_u(Q)$-symmetry (which in particular implies that detailed balance holds, in view of Proposition \ref{propGlongQ}) then each $\GH_m$ carries an irreducible \emph{unitary} representation precisely when the inner product is the one defined by $Q_m$, which is the inner product on $\GH_m$ as a subspace of the GNS space $\Hi_{\rho_0}$. 

Consider now the Stinespring spaces $\overline{\GH}_m$ of the time-reversal channel $\bar{\Phi}$. We can identify $\overline{\GH}_m$ with the complex conjugate of $\GH_m$. Again the inner product has to be changed to make the $\G$-representation on $\overline{\GH}_m$ unitary; this time we have to use $(Q_m^{-1})^t$. In another paper we shall be concerned with determining which \emph{inequivalent} representations of $\G$ occur as a space $\GH_m$ or $\overline{\GH}_m$ or as a ``product" of such spaces. The point is that we can identify $\Bi(\GH_m)$ with the space of trajectories $K_\mathbf{j}^*K_\mathbf{k}$ for the forward evolution, while $\Bi(\overline{\GH}_m)$ contains the trajectories $\bar{K}_\mathbf{j}^*\bar{K}_\mathbf{k}$ of the backward evolution. The Feynman path integral over inequivalent trajectories becomes a sum in which each term is weighted in terms of a $\G$-invariant ``Laplacian" on the ``manifold" of trajectories. It is then an interesting question in what cases the backward trajectories are inequivalent to the forward ones, in this sense of $\G$-representations. If every backward path can be reconstructed from a combination of forward ones in a $\G$-invariant way, it is sufficient to include the forward trajectories in the Feynman sum, since they already form a complete basis of eigenvectors for the Laplacian.

The above paragraph is the motivation of our notion of ``time-reversal" (Definition \ref{TRinvdef}). If $\Phi$ is $\G$-symmetric for some compact matrix quantum group $\G$ then time-reversal invariance of $\Phi$ is equivalent to $\G\subset B_u(F)$ (i.e. $\G$ has self-conjugate defining representation). Interestingly, in case of time reversal the Kraus operators \emph{have} to satisfy some commutation relations, and in this way we see a relation between symmetry and statistics. 

\begin{Remark} For a closed quantum system, time reversal is simply given by an anti-unitary operator $\Theta$ with $\Theta^2=\pm\bone$, so that the time reversal of an operator $A$ is given by $\Theta A\Theta^{-1}$. In particular this is true for the unitary time-evolution operator $A=W$. An especially nice situation is when $\Theta$ is simply given by complex conjugation. If the Hilbert space is of the form $\Hi_0\otimes\GH$ with $\GH$ finite-dimensional, one may want to take $\Theta$ as ``complex conjugation" of matrices, i.e. as entrywise $*$-involution. However, as we have observed, $\Theta W\Theta^{-1}$ is in that case not a unitary in general. It therefore makes sense to let $\Theta\psi=\overline{(\bone\otimes F^*)\psi}$ instead (for vectors $\psi$, where $F$ is as in Definition \ref{TRinvdef} and the bar denotes complex conjugation), which would satisfy $\Theta^2=\pm\bone$ if $F^2=\pm\bone$; this can always be arranged (up to isomorphism) if $\Phi$ has $B_u(F)$ symmetry \cite[Eqs. (5.3),(5.4)]{BDV1}.

A recent paper \cite{AZZ1} refers to the channel corresponding to $\Theta W\Theta^{-1}$ as the ``environmental time reversal" of $\Phi$.
\end{Remark}

For commuting Kraus operators we can only have a classical symmetry group $G$. Here $Q=\bone$ is the only possibility for detailed balance, so the channel is just the Hilbert-Schmidt adjoint of its time-reversal, $\bar{K}_j=K_j^*$. The backward process can be reconstructed perfectly from the forward one and, more sensibly, the forward process can be derived from looking backwards. In that sense we have perfect time reversal. Since the Kraus operators correspond to the first row of the definition representation of $G$, they generate the algebra $C(G/K)$ of continuous functions on a $G$-homogeneous space $G/K$. We see that in the context of quantum detailed balance, classical symmetry can describe a very limited type of systems.

The simple structure of the representation theory of $\SU(2)$ makes its occurrence in quantum mechanics very fortunate. In comparison, $\SU(n)$ for $n\geq 3$ is much more complicated and the quantum version $\SU_q(n)$ is rather horrible. As mentioned in Remark \ref{GlebschGordanrem}, the ``higher-dimensional quantum $\SU(2)$ group" $B_u(F)$ has the same Clebsch-Gordan structure of its representations as $\SU(2)$. It is therefore very interesting that $B_u(F)$ seems to have a straightforward interpretation as time-reversal symmetry.

The rather weak notion of time-reversal in Definition \ref{UCMQGdef} requires that there is at least some \emph{self-similarity} in the environment interaction, in contrast to the ``white noise" $A_u(Q)$. The similarity between $B_u(F)$ and $\SU(2)$ suggests that we should view the algebra $\GT^{(0)}$ generated by elements of the form $K_\mathbf{j}K_\mathbf{k}^*$ with $|\mathbf{j}|=|\mathbf{k}|$ as defining a deformed sphere, since if $\Phi$ has $\SU(2)$ symmetry then $\GT^{(0)}\cong C(\Sb^2)$, while for $\SU_q(2)$ we have $\GT^{(0)}\cong C(\Sb^2_q)$. The deformed sphere $\Sb_q^2$ behaves also like a fractal \cite{EIS1}. Either picture suggests that the time evolution is not perfectly repeating itself, and it is $Q\ne \bone$ which breaks perfect time-reversal symmetry. In summary:
\begin{enumerate}[$\bullet$]
\item{The symmetry of $A_u(Q)$ is the minimal one for detailed balance to hold.}
\item{The symmetry $B_u(F)\subset A_u(Q)$ may be regarded as a time-reversal symmetry, and this requires non-free commutation relations among the Kraus operators.}
\item{Any other algebraic relations may spoil quantum-group symmetry or result in a different symmetry $\G\subset A_u(Q)$.}
\end{enumerate}

\section{References}

\bibitem[AcLu1]{AcLu1} Accardi L, Lu YG. The Wigner semi-circle law in quantum electro dynamics. Commun. Math. Phys. 180, 605-632 (1996).

\bibitem[Agar1]{Agar1} Agarwal GS. Open quantum markovian systems and the microreversibility. Z. Physik 258, Vol pp. 409-422 (1973).

\bibitem[ALMZ1]{ALMZ1} Albash T, Lidar DA, Marvian M, Zanardi P. Fluctuation theorems for quantum processes. Phys Rev E. Vol 88, 032146 (2013).

\bibitem[AlFa1]{AlFa1} Alicki R, Fannes M. Quantum dynamical systems. Oxford University Press, Oxford (2001).

\bibitem[AlLe1]{AlLe1} Alicki R, Lendi K. Quantum dynamical semigroups and applications. Lecture Notes in Phys. Vol 717, Springer-Verlag Berlin Heidelberg (2007).

\bibitem[An1]{An1} Andersson A. Operator deformations in quantum measurement theory. Lett. Math. Phys. Vol 104, Issue 4, pp. 415-430 (2014). 

\bibitem[An5]{An5} Andersson A. Dequantization via quantum channels.  arXiv: 1506.01453 (2015).

\bibitem[An6]{An6} Andersson A. Berezin quantization of noncommutative projective varieties. arXiv: 1506.01454 (2015).

\bibitem[ADF1]{ADF1} Angelova MN, Dobrev VK, Frank A. Revisiting the quantum group symmetry of diatomic molecules. The European Physical Journal D-Atomic, Molecular, Optical and Plasma Physics. Vol 31, Issue 1, pp. 27-37 (2004).

\bibitem[AnFr2]{AnFr2} Angelova MN, Frank A. Algebraic approach to thermodynamic properties of diatomic molecules. Physics of Atomic Nuclei. Vol 68, Issue 10, pp. 1625-1633 (2005). 

\bibitem[Ar3]{Ar3} Arveson W. The index of a quantum dynamical semigroup. Journal of functional analysis. Vol 146, Issue 2, pp. 557-588 (1996).

\bibitem[Ar7]{Ar7} Arveson W. Subalgebras of $C^*$--algebras. Acta Math. Vol 123, pp. 141-224 (1969).

\bibitem[AZZ1]{AZZ1} Aurell E, Zakrzewski J, Życzkowski K. Time reversals of irreversible quantum maps. arXiv:1505.02259 (2015).

\bibitem[Ban4]{Ban4} Banica T. Le groupe quantique compact libre $\Un(n)$. Commun. Math. Phys. Vol 190, pp. 143–172 (1997).

\bibitem[Ban6]{Ban6} Banica T. Representations of compact quantum groups and subfactors. Journal für die reine und angewandte Mathematik (Crelles Journal), Issue 509, pp. 167-198 (1999). 

\bibitem[BDV1]{BDV1} Bichon J, De Rijdt A, Vaes S. Ergodic coactions with large multiplicity and monoidal equivalence of quantum groups. Comm. Math. Phys. Vol 262, pp. 703-728 (2006).

\bibitem[BoDa1]{BoDa1} Bonatsos D, Daskaloyannis C. Quantum groups and their applications in nuclear physics. Progress in Particle and Nuclear Physics. Vol 43, pp. 537-618 (1999).

\bibitem[BKLRRT]{BKLRRT} Bonatsos D, Karoussos N, Lenis D, Raychev PP, Roussev RP, Terziev PA. Unified description of magic numbers of metal clusters in terms of the three-dimensional $q$-deformed harmonic oscillator. Physical Review A. Vol 62. Issue 1, 013203 (2000)

\bibitem[Bon1]{Bon1} Bonatsos D. Are $q$-bosons suitable for the description of correlated fermion pairs? Journal of Physics A: Mathematical and General. Vol 25, Issue 3, L101 (1992).

\bibitem[BrRo2]{BrRo2}  Bratteli O, Robinson D. Operator algebras and quantum statistical mechanics 2: Equilibrium states. Models in quantum statistical mechanics. Springer Science and Business Media (1997).

\bibitem[BP]{BP} Breuer HP, Petruccione F. The theory of open quantum systems. Oxford University Press, Oxford (2003).

\bibitem[Choi1]{Choi1} Choi MD. Completely positive linear maps on complex matrices. Lin. Alg.Appl. Vol 10, pp. 285-290 (1975).

\bibitem[Crook1]{Crook1} Crooks GE. Quantum operation time reversal. Physical Review A. Vol 77, Issue 3, p. 034101. (2008).

\bibitem[Crook4]{Crook4} Crooks GE. Entropy production fluctuation theorem and the nonequilibrium work relation for free energy differences. Physical Review E. Vol 60, Issue 3, p. 2721 (1999).

\bibitem[EIS1]{EIS1} Eckstein M, Iochum B, Sitarz A. Heat trace and spectral action on the standard Podleś sphere. Communications in Mathematical Physics. Vol 332, Issue 2, pp. 627-668 (2014).

\bibitem[FaRe1]{FaRe1} Fagnola F, Rebolledo R. Entropy production for quantum Markov semigroups. arXiv:1212.1366 (2012).

\bibitem[FIL1]{FIL1}  Frank A, Iachello F, Lemus R. Algebraic methods for molecular electronic spectra. Chemical physics letters. Vol 131, Issue 4, pp. 380-383 (1986).

\bibitem[Guhr1]{Guhr1} Guhr T, Müller-Groeling, A, Weidenmüller HA. Random-matrix theories in quantum physics: common concepts. Physics Reports. Vol 299, Issue 4, pp. 189-425 (1998).

\bibitem[Janz1]{Janz1} Janzing D. Decomposition of time-covariant operations on quantum systems with continuous and/or
discrete energy spectrum. Journal of mathematical physics. Vol 46, Issue 12, p. 122107 (2005).

\bibitem[Jarz1]{Jarz1} Jarzynski C. Nonequilibrium equality for free energy differences. Phys. Rev. Lett. Vol 78, p. 2690 (1997).

\bibitem[Jarz2]{Jarz2} Jarzynski C. Equalities and inequalities: irreversibility and the second law of thermodynamics at the nanoscale. Annu. Rev. Condens. Matter Phys. Vol 2, pp. 329-351 (2011). 

\bibitem[KiNe1]{KiNe1}  Kibler M, Négadi T. On the $q$-analogue of the hydrogen atom. Journal of Physics A: Mathematical and General. Vol 24, Issue 22, 5283 (1991).

\bibitem[KiNe3]{KiNe3} Kibler M, Négadi T. On quantum groups and their potential use in mathematical chemistry. Journal of Mathematical Chemistry. Vol 11, Issue 1, pp. 13-25 (1992).

\bibitem[KlS]{KlS} Klimyk AU, Schmüdgen K. Quantum groups and their representations. Vol. 552, Springer, Berlin (1997).

\bibitem[KFGV]{KFGV} Kossakowski A, Frigerio A, Gorini V, Verri M. Quantum detailed balance and KMS condition. Comm. Math. Phys. Vol 57, pp. 97-110 (1977).

\bibitem[Kum1]{Kum1}  Kümmerer B. Quantum Markov processes and applications in physics. Lecture Notes in Mathematics. Vol 1866, Springer, Berlin, pp. 259-328 (2005).

\bibitem[LOMI]{LOMI} Liu YX, Özdemir ŞK, Miranowicz A, Imoto N. Kraus representation of a damped harmonic oscillator and its application. Physical Review A. Vol 70, Issue 4, 042308  (2004). 

\bibitem[MHP1]{MHP1} Manzano G, Horowitz JM, Parrondo JMR. Nonequilibrium potential and fluctuation theorems for quantum maps. arXiv:1505.04201v1 (2015).

\bibitem[MeWo1]{MeWo1} Mendl CB, Wolf MM. Unital quantum channels - convex structure and revivals of Birkhoff's theorem. Comm. Math. Phys. Vol 289, pp. 1057-1096 (2009).

\bibitem[Monn1]{Monn1} Monnai T. Microscopic reversibility of quantum open systems. Journal of Physics A: Mathematical and Theoretical. Vol 45, Issue 12, p. 125001 (2012).

\bibitem[NCSB1]{NCSB1} Nielsen MA, Caves CM, Schumacher B, Barnum H. Information-theoretic approach to quantum error correction and reversible measurement. Proceedings of the Royal Society of London. Series A: Mathematical, Physical and Engineering Sciences. Vol 454, Issue 1969, pp. 277-304 (1998).

\bibitem[Ojim1]{Ojim1} Ojima I. Micro-macro duality in quantum physics. arXiv: math-ph/0502038 (2005).

\bibitem[Pain1]{Pain1} Pain JC. Random-matrix theory and complex atomic spectra. Chinese journal of physics. Vol 50, Issue 4, p. 523 (2012).

\bibitem[Rast1]{Rast1} Rastegin AE. Non-equilibrium equalities with unital quantum channels. Journal of Statistical Mechanics: Theory and Experiment. Issue 06, P06016 (2013).

\bibitem[RaZy1]{RaZy1} Rastegin AE, Życzkowski K. Jarzynski equality for quantum stochastic maps. Physical Review E. Vol 89, Issue 1, 012127 (2014).

\bibitem[Rayc1]{Rayc1} Raychev P. Quantum groups: Application to nuclear and molecular spectroscopy. Advances in quantum chemistry. Vol 26, pp. 239-357 (1995).

\bibitem[Schu1]{Schu1} Schumacher B. Sending entanglement through noisy quantum channels. Physical Review A. Vol 54, Issue 4, p. 2614 (1996).

\bibitem[ShSo1]{ShSo1} Shalit OM, Solel B. Subproduct systems. Documenta Mathematica. Vol 14, pp. 801-868 (2009).

\bibitem[Stor1]{Stor1} Størmer E. The analogue of Choi matrices for a class of linear maps on Von Neumann algebras. arXiv: 1412.8598 (2014).

\bibitem[Timm1]{Timm1} Timmermann T. An invitation to quantum groups and duality: from Hopf algebras to multiplicative unitaries and beyond. European Mathematical Society (2008). 

\bibitem[VaDW]{VaDW} Van Daele A, Wang SZ. Universal quantum groups. International J. Math Vol7, Issue 2, pp. 255-264 (1996).

\bibitem[Wang3]{Wang3} Wang SZ. Structure and isomorphism classification of compact quantum groups $A_u (Q)$ and $B_u (Q)$. arXiv: math/9807095v2 (2000).

\bibitem[Wign1]{Wign1} Wigner E. On the distribution of the roots of certain symmetric matrices. Ann. Math. Vol 67, pp. 325-327 (1958).

\bibitem[Wor1]{Wor1} Woronowicz SL. Compact matrix pseudogroups. Comm. Math. Phys. Vol 111, pp. 613-665 (1987).

\bibitem[Wor3]{Wor3} Woronowicz SL. A remark on compact matrix quantum groups. Lett. Math. Phys. Vol 21, pp. 35-39 (1991).

\end{document}